\documentclass[journal]{IEEEtran}
\usepackage{mathrsfs}
\usepackage{amssymb}
\usepackage{amsmath}
\usepackage{amsthm}
\usepackage{graphicx}
\usepackage{subfigure}
\usepackage{cite}
\usepackage{enumerate}
\usepackage[ruled]{algorithm2e}
\usepackage{color, soul}
\usepackage{verbatim}
\usepackage{amsfonts}
\usepackage{array}
\usepackage[marginal]{footmisc}
\usepackage{diagbox}
\usepackage{bm}
\usepackage{epstopdf}

\begin{document}

\newtheorem{definition}{\bf ~~Definition}
\newtheorem{observation}{\bf ~~Observation}
\newtheorem{lemma}{\bf ~~Lemma}
\newtheorem{proposition}{\bf ~~Proposition}
\newtheorem{remark}{\bf ~~Remark}
\newcommand{\tabincell}[2]{\begin{tabular}{@{}#1@{}}#2\end{tabular}}

\title{Device-to-Device Communications Underlaying Cellular Networks: To Use Unlicensed Spectrum or Not?}

\author{
\IEEEauthorblockN{
{Fanyi Wu}, \IEEEmembership{Student Member, IEEE},
{Hongliang Zhang}, \IEEEmembership{Student Member, IEEE},
{Boya Di}, \IEEEmembership{Student Member, IEEE},
{Jianjun Wu},
{and Lingyang Song}, \IEEEmembership{Fellow, IEEE}}\\

\vspace{-0.5cm}

\thanks{Manuscript received December 22, 2018; revised April 17, 2019; accepted May 12, 2019. This work was supported by the National Natural Science Foundation of China under Grant 61625101. Part of this paper was presented at IEEE ICC, Shanghai, China, May~2019~\cite{FHBJL-2019}. The associate editor coordinating the review of this paper and approving it for publication was D.~Niyato.}

\thanks{The authors are with Department of Electronics Engineering, Peking University, Beijing, China (email: fanyi.wu@pku.edu.cn, hongliang.zhang@pku.edu.cn, diboya@pku.edu.cn, just@pku.edu.cn, lingyang.song@pku.edu.cn)}
}
\maketitle

\vspace{-0.5cm}

\begin{abstract}
In this paper, we consider device-to-device~(D2D) communications as an underlay to the cellular networks over both licensed and unlicensed spectrum, where Long Term Evolution~(LTE) users utilize the spectrum orthogonally while D2D users share the spectrum with LTE users. In the system, each LTE and D2D user can access the licensed or unlicensed band for communications. To maximize the total throughput of the system, we leverage stochastic geometry to derive the throughput for each kind of users by modeling the deployment of users as Poisson point processes~(PPPs), and investigate the spectrum access problem for these users. Since the problem is NP-hard, we propose a sequential quadratic programming~(SQP) based algorithm to obtain the corresponding suboptimal solutions. Theoretically, we evaluate the system performance by analyzing the throughput regions. Simulation results validate the accuracy of the geometric analysis and verify the effectiveness of the proposed algorithm.
\end{abstract}


\section{Introduction}%
\label{Introduction}

The rapid growth of mobile devices and Internet-based applications has led to the explosive increase in mobile data traffic. It has been predicted that the mobile data demand will grow over 1000-fold in the next decade, also known as ``1000x mobile data challenge''~\cite{WP_Rising}. To tackle this challenge, many technologies have been developed to enhance the spectral efficiency. Device-to-device~(D2D) communication is one of the promising solutions in this regard~\cite{AQV-2014,KMCCK-2009}. Specifically, D2D communications allow two proximal users to set up direct link bypassing the BS and to share the licensed spectrum with the cellular networks. Due to the proximity~\cite{GEGSNGZ-2012} and underlay property~\cite{HLZ-2016}, D2D communications can effectively improve can effectively improve the throughput as well as the energy efficiency~\cite{RYGNC-2016,JJCWW-2016} in the network.

However, the limited licensed spectrum may not be sufficient to support D2D communications, especially in the hotspot scenarios. In order to further improve system throughput, D2D communications underlaying cellular networks over the unlicensed spectrum are proposed, referred to as D2D-Unlicensed~(D2D-U) communications~\cite{HYL-2017}, which can be facilitated by the existing LTE-Unlicensed~(LTE-U) technologies~\cite{WP_Extending,YWHLSXJ-2016}. In the D2D-U networks, users can share the unlicensed spectrum with the existing Wi-Fi systems fairly and harmoniously via the Listen-Before-Talk~(LBT) mechanism~\cite{RMLZXL-2015}, which has been adopted in the LTE-licensed assisted access~(LTE-LAA) technology~\cite{BJYJ-2017}. Unlike Wi-Fi Direct, which allows two users to communicate directly by one user serving as an access point~(AP)~\cite{CHL-2015}, D2D-U communications require the assistance and control from the base station~(BS). With the network-assisted control, D2D-U communications can provide reliable services for users over the unlicensed spectrum.

In a D2D-U system, each LTE/D2D user can access the networks via either the licensed or unlicensed spectrum but the spectrum access method is different over these two bands. To be specific, LTE and D2D users utilize the licensed spectrum based on the orthogonal frequency division multiple access~(OFDMA) technique, while access the unlicensed spectrum by adopting the LBT mechanism\footnote{Actually, the spectrum access of LTE and D2D users operates similar to the existing cognitive radio~(CR) technique~\cite{YSRYC-2013}.}. To coordinate these two methods and improve the capacity of the networks, it is necessary to investigate the spectrum access problems for LTE and D2D users. In this paper, we consider a unified D2D-U network where LTE, D2D and Wi-Fi users coexist over both licensed and unlicensed spectrum. D2D users work as an underlay to LTE users in both licensed and unlicensed bands. To fairly and harmoniously coexist with Wi-Fi users, both LTE and D2D users utilize a Load-Based-Equipment~(LBE) based LBT mechanism~\cite{Netmanias} for the unlicensed channel access.

Nevertheless, it is challenging to build a unified framework to investigate the spectrum access problem in such a D2D-U network. First, a suitable mathematical model is required to theoretically analyze the performances of different users in the system and formulate the opportunistic feature of channel access over the unlicensed spectrum. Second, the spectrum access of LTE and D2D users over licensed and unlicensed spectrum is difficult to be investigated jointly, since the interference among LTE, D2D, and Wi-Fi users is complicated due to the underlay property of D2D users. To tackle the first issue, we leverage stochastic geometry to model the locations of LTE users, D2D transmitters~(TXs) and Wi-Fi APs as independent homogeneous Poisson point processes~(PPPs) with diverse densities. Besides, we formulate the unlicensed channel access of LTE, D2D and Wi-Fi users as hard core point processes~(HCPPs). To cope with the second issue, we investigate the spectrum access problem by optimizing the densities of different users to maximize the total throughput of the system. Since this problem is NP-hard, we design a sequential quadratic programming~(SQP) based spectrum access algorithm to obtain the suboptimal solutions iteratively.

In the literature, several works have investigated on the resource allocation~\cite{HYL-2017,RGFZ-2016} and performance analysis~\cite{BMDWD-2017,HWS-2016} in the D2D-U networks. Authors in~\cite{HYL-2017} proposed a resource allocation algorithm for D2D-U system using matching theory. In~\cite{RGFZ-2016}, a joint mode selection and resource allocation algorithm was proposed to minimize the overall interference in a D2D-U network, in which the quality of service~(QoS) requirements were considered. In~\cite{BMDWD-2017}, the performance comparison of D2D-U networks was discussed with different unlicensed access techniques for D2D users. The waiting probability, time delay and network capacity for D2D-U networks were analyzed based on a varying traffic model in~\cite{HWS-2016}. Different from the aforementioned works, we present a unified analytical framework to jointly investigate the spectrum access for both LTE and D2D users in the D2D-U network. Besides, there also exist several works on the spectrum access problem in the LTE-U networks~\cite{RRYS-2017,FEEMR-2015}, which only focused on the spectrum access among LTE users. In~\cite{RRYS-2017,FEEMR-2015}, LTE users accessed the spectrum based on the expected payoff~\cite{RRYS-2017} and the utility of users~\cite{FEEMR-2015}, respectively. Unlike the LTE-U systems, the interference among LTE and D2D users makes the spectrum access problem more complicated in the D2D-U systems. Moreover, the existing works on the stochastic geometry analysis in D2D underlaying cellular networks only focused on the licensed spectrum access. For instance, authors in~\cite{MBPF-2016} investigated on the power distribution in a D2D underlaying cellular network over the licensed spectrum using stochastic geometry. Compared with the licensed spectrum access, analyzing the unlicensed spectrum access using stochastic geometry is more difficult, since the opportunistic feature of unlicensed spectrum access needs to be formulated.

The main contributions of this paper are summarized as follows:
\begin{itemize}
\item We study a unified D2D-U network consisting of LTE, D2D, and Wi-Fi users, and then analytically derive their throughput over both licensed and unlicensed spectrum by leveraging stochastic geometry.
\item We jointly investigate the spectrum access problem for LTE and D2D users by maximizing the total throughput in the system, and then propose a SQP-based algorithm to obtain the suboptimal solutions.
\item We characterize the throughput regions to analyze the system performance and illustrate the spectrum access issue.
\end{itemize}

The rest of the paper is organized as follows. In Section \ref{System Model}, we introduce the system model for the coexistence among LTE, D2D and Wi-Fi users, and then present the spectrum sharing scheme of the network. In Section~\ref{Throughput Analysis}, we derive the throughput for each kind of users by geometric analysis. Then, we formulate the spectrum access problem and design a SQP-based algorithm to solve it in Section~\ref{Problem Formulation and Algorithm Design}. We analyze the convergence and complexity of the algorithm, and then illustrate the system performance via the throughput regions in Section~\ref{Performance Analysis}. Simulation results are presented in Section~\ref{Simulation Result}. Finally, Section~\ref{Conclusion} concludes this paper.

\section{System Model}%
\label{System Model}
In this section, we first describe the system consisting of LTE, D2D, and Wi-Fi users. Then, the spectrum sharing schemes over licensed and unlicensed spectrum are elaborated, respectively. For clarity, in following parts of this paper, we define cellular and D2D users over the licensed spectrum as \emph{LTE} and \emph{D2D} users, while those over the unlicensed spectrum are referred to \emph{LTE-U} and \emph{D2D-U} users, respectively.

\subsection{System Description}%

\begin{figure}[!t]
\centering
\vspace{-1mm}
\includegraphics[width=3.2in]{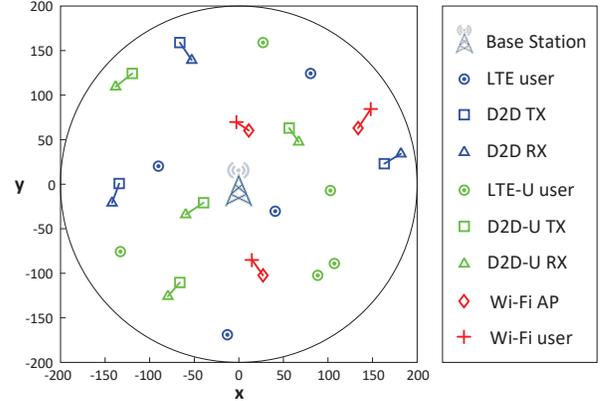}
\vspace{-3mm}
\caption{System model for LTE, D2D and Wi-Fi users coexistence in both licensed and unlicensed spectrum.}
\vspace{-5mm}
\label{scenario}
\end{figure}

As shown in Fig.~\ref{scenario}, we consider an OFDMA uplink single-cell system\footnote{Actually, our work can also be applied in the downlink scenario.}. The cell is modeled as a circular area with the BS at the center and the radius as $r_{cell}$. LTE users and D2D users utilize the licensed spectrum with the total bandwidth $B_l$, and LTE-U users and D2D-U utilize the unlicensed spectrum with the total bandwidth $B_u$. We assume that LTE and LTE-U users are distributed on $\mathbb{R}^2$ within the cell according to independent homogeneous PPPs with the densities of $\lambda_C$ and $\lambda_{CU}$, denoted by $\psi_C$ and $\psi_{CU}$, respectively. Similarly, we model the deployment of TXs of D2D users and D2D-U users as other independent homogeneous PPPs with the densities of $\lambda_D$ and $\lambda_{DU}$, denoted by $\psi_D$ and $\psi_{DU}$, respectively. The receiver~(RX) of each D2D/D2D-U user is uniformly distributed in a circular area centered on its TX with radius $L_d$, where $L_d \ll r_{cell}$. Besides, Wi-Fi APs are randomly distributed on the plane $\mathbb{R}^2$ within the cell following homogeneous PPP with the density of $\lambda_{W}$, denoted by $\psi_{W}$. For each AP, its associated Wi-Fi users are uniformly distributed in a circular area centered on the AP with the radius $L_w$, where $L_w \ll r_{cell}$.

In addition, we assume that the transmit power of LTE/LTE-U users is fixed on $P_C$, and the transmit power of D2D/D2D-U users is fixed on $P_D$, where $P_D < P_C$. Besides, the transmit power of Wi-Fi users is set as $P_W$. We consider a free space propagation path-loss model~\cite{PBHKL-2018}. Hence, the channel gain between two nodes $m$ and $n$ can be calculated as $g_{m,n} = d_{m,n}^{ - \alpha }h$, where $d_{m,n}$ indicates the distance between $m$ and $n$, $\alpha$ is the path loss exponent, and $h$ is the fading coefficient due to the small-scale fading. Assuming a frequency non-selective block fading channel, $h$ follows an exponential distribution\cite{SD-2017}, i.e., $h \sim \exp(1)$. The thermal noise at each user satisfies independent Gaussian distribution with zero mean and the same variance $\sigma^2$.

\subsection{Spectrum Sharing Schemes}
In our system, the spectrum access of LTE/D2D users over both licensed and unlicensed spectrum is controlled by the BS, which is different from the traditional ad hoc network~\cite{EC-1999}. For fair and harmonious spectrum access, we design the spectrum sharing schemes over licensed and unlicensed spectrum as follow\footnote{To support the mixed data transmissions over both the licensed and the unlicensed spectrum, the hardware of each D2D pair should be compatible with the $5$ GHz LTE-enabled hardware in the LTE-LAA technique~\cite{SLAAUS}. Specifically, each D2D pair should be equipped with LBE-based LBT module, dynamic frequency selection~(DFS) module, carrier aggregation~(CA) module, spectrum access switch, single filter, etc.}.

\subsubsection{Licensed Spectrum}
In the licensed spectrum, the frequency band is equally divided into subchannels to support OFDMA transmissions for both LTE and D2D users. More explicitly, we equally divide the licensed band $B_l$ into $K_l$ subchannels with uniform bandwidth of $B_l^{sub}=\frac{B_l}{K_l}$, denoted by $\mathcal {K}_l = \{ 1, 2, ... , K_l \}$, where $K_l={{\lambda _C}\pi r_{cell}^2}$ is the number of LTE users. As such, each LTE user is assigned to one orthogonal subchannel by the BS for data transmission. Each D2D user, working as an underlay, is allowed to share the spectrum with LTE users \cite{RJH-2015}. For fairness, we assume that each D2D user can occupy at most one subchannel for communication, and the BS selects a subchannel from $\mathcal {K}_l$ to support the D2D user.

\subsubsection{Unlicensed Spectrum}
The existing Wi-Fi systems adopt the carrier sense multiple access with collision avoidance~(CSMA/CA) protocol~\cite{MMAWSM-2013} to share the unlicensed band. Specifically, a Wi-Fi AP monitors the intended channel before transmission. When the channel is idle for a period of time, an exponential random back-off procedure begins, i.e., a back-off integral counter is generated, and then the AP keeps sensing the channel. As soon as the channel is judged idle for a time slot, the counter decreases by one. However, if the channel is sensed busy, the counter will freeze until the channel is sensed idle again. When the counter decreases to zero, the Wi-Fi user will occupy the whole unlicensed channel for transmission~\cite{KWY-2014}.

In order to coexist with Wi-Fi users in a friendly manner, we assume that both LTE-U and D2D-U users utilize the LBE-based LBT mechanism~\cite{Netmanias}, which operates quite similar to the CSMA/CA protocol of Wi-Fi system\footnote{The LBE-based LBT mechanism is a variant of the CSMA/CA protocol. In this mechanism, the users performs clear channel assessment~(CCA) whenever there exists data to be transmitted. If the channel is clear, the data can be transmitted immediately, otherwise, the data attempts to be transmitted after back-off during the extended CCA.}. However, different from the CSMA/CA protocol of the Wi-Fi system which adopts exponential back-off, the LBE-based LBT mechanism implements fixed linear back-off with fixed size of contention windows. To further improve the spectrum efficiency, we divide the unlicensed band $B_u$ into $K_u$ subchannels with uniform bandwidth $B_u^{sub} = \frac{B_u}{K_u}$, denoted by $\mathcal {K}_u = \{ 1, 2, ... , K_u \}$, to support the concurrent transmissions for LTE-U and D2D-U users. Here, $K_u = {{\lambda _{CU}}\pi r_{cell}^2}$ is the number of LTE-U users. Before transmissions, the BS first allocates the subchannels in $\mathcal {K}_u$ to the LTE-U and D2D-U users. Since different LTE-U users in the same subchannel may bring severe interference, each LTE-U user is assigned to one orthogonal subchannel to avoid the interference. As for D2D-U users, each user is allocated to one subchannel randomly for fairness. After that, each user monitors the subchannel and decides whether to occupy it or not. Once the subchannel is sensed idle, each user backs off for a random period, which is similar to the Wi-Fi system.

It is worth to mention that the data transmissions of all users are over the data channels while the signaling is transmitted on the control channels. Since data and control channels are operated on different bands, the signaling will not affect the performance of users. In our system, the signaling is performed on extra control channels, and thus, we only focus on the data transmissions over the data channels. However, we still present the signaling analysis on LTE-U and D2D-U users in the unlicensed spectrum as follow. We define that $M_{si}$ messages are required to inform the BS the subchannel information sensed by a LTE-U/D2D-U user, and $M_{sa}$ messages are needed for the BS to notify a LTE-U/D2D-U user about the subchannel allocation. Since each LTE-U/D2D-U user needs to report the sensing result over single allocated subchannel at a time, $M_{si}(\lambda_{CU}+\lambda_{DU})\pi r_{cell}^2$ messages are required to inform the BS. Besides, for subchannel allocation, the BS notifies all LTE-U and D2D-U users by sending $M_{sa}(\lambda_{CU}+\lambda_{DU})\pi r_{cell}^2$ messages.

\section{Throughput Analysis}
\label{Throughput Analysis}

In this section, we analyze the throughput for different users, including LTE and D2D users over the licensed spectrum, and LTE-U, D2D-U, and Wi-Fi users over the unlicensed spectrum.

\subsection{Throughput Analysis in the Licensed Spectrum}
In the licensed spectrum, due to the underlay property, the interplay between LTE and D2D users needs to be considered.
\subsubsection{LTE Users}
Define the distance between LTE user $u$ and the BS as $d_{u,b}$. Based on the property of PPP, $d_{u,b}$ follows a uniform distribution whose probability distribution function~(PDF) is given by ${f_{{d_{u,b}}}}(r) = {2r}/{r_{cell}^2},~ 0 \le r \le {r_{cell}}$. Since different LTE users utilize orthogonal subchannels for data transmissions, the interference to LTE user $u$ only comes from D2D users sharing the same subchannel. For simplicity, define the interference function $\mathcal {I}(z,P,\psi) = \sum_{i \in \psi } P d_{i,z}^{ - \alpha }h$. Then, the interference from D2D users can be express as $I_D = \mathcal {I}(b,P_D,\psi'_D)$, where $\psi'_D$ is the set of all co-channel D2D users of LTE user $u$, which is a thinning homogeneous PPP of $\psi_D$ with the density of $\lambda'_D = \frac{\lambda_D}{K_l}$. The signal-to-interference-plus-noise ratio~(SINR) of LTE user $u$ can be expressed as
\begin{equation}
\gamma _{C,u} = \frac{{{P_C}d_{u,b}^{ - \alpha }h}}{{{\sigma ^2} + {I_D}}}.
\end{equation}

Therefore, given the bandwidth $B_l^{sub}$, the throughput of LTE user $u$ can be calculated by
\begin{equation}\label{equ_Rc}
  R_{C,u} = B_l^{sub}\mathbb{E}\left\{ {\log_2 \left( {1 + \gamma _{C,u} } \right)} \right\},
\end{equation}
where $\mathbb{E}\left\{\cdot\right\}$ is the mathematical expectation. Note that the interference to different LTE users follows a uniform distribution, and thus, the throughput of different LTE users is the same. For brevity, we define $R_C = R_{C,u}$ as the throughput of an LTE user.

\begin{proposition}\label{Rc}
The throughput of an LTE user over the licensed spectrum is given by
\begin{equation}\label{Rc_exp}
R_C = {B_l^{sub}} \vartheta \left(P_C,  \mathcal{Q}\left(\lambda'_D, P_D, 0\right),  \delta(y), f_{d_{u,b}}(r) \right),
\end{equation}
where ${\lambda'_D} = \frac{\lambda_D}{K_l}$, $\delta(\cdot)$ is the Dirac function, $\vartheta(\cdot)$, $\mathcal{Q}(\cdot)$ and $\mathcal {H}(\cdot)$ are given by:
\begin{align}
&\vartheta \left(P,Q,g(y),f(r)\right) = \nonumber\\
&\qquad\int_{0}^\infty  \int_{-\infty}^\infty \int_{-\infty}^\infty  \frac{e^{ -s{\sigma ^2}}Q}{{\left( {s + {P^{ - 1}}{r^\alpha }} \right)\ln 2}}  g(y)f(r)dydrds,\label{function_eta} \\
&\mathcal{Q}\left(\lambda,P,l\right) = {e^{ - \lambda \int_0^{2\pi } {\int_0^{{r_{cell}}} {\frac{1}{{1 + {{\left( {sP} \right)}^{ - 1}}{\mathcal{H}^{\alpha /2}}(x,\theta ,l)}}x} dxd\theta } }},\label{function_Q} \\
&\mathcal {H}\left( x,\theta,l \right) = x^2 + l^2 - 2xl\cos\theta.\label{function_H}
\end{align}
\end{proposition}
\begin{proof}
See Appendix \ref{Proof1}.
\end{proof}

\subsubsection{D2D Users}
Define the distance between the TX and the RX of D2D user $v$ as $d_{v,v}$. As a RX is randomly distributed in a circular area in which the center is the TX and the radius is $L_d$, the PDF for $d_{v,v}$ can be given by ${f_{d_{v,v}}}(r) = {2r}/{L_d^2},~0 \le r \le {L_d}$. Since each subchannel is utilized by one LTE user, the RX of D2D user $v$ is only interfered by one LTE user, denoted by $u$. Besides, D2D user $v$ is also interfered by other D2D users coexisting in the same subchannel. Since $r_{cell} \gg L_d$, the distance from the interfering node to the RX of D2D user $v$ can be approximated to the distance between the interfering node and the TX of D2D user $v$. Then, we can express the interference from LTE and D2D users as $I_C = \mathcal {I}(v,P_C,\{u\})$ and ${I_D} = \mathcal {I}(v,P_D,\psi'_D \setminus v)$, and the SINR of D2D user $v$ can be given by
\begin{equation}
\gamma_{D,v} = \frac{{{P_D}d_{v,v}^{ - \alpha }h}}{{{\sigma ^2} + {I_C} + {I_D}}}.
\end{equation}

Therefore, given the bandwidth $B_l^{sub}$, the throughput of D2D user $v$ can be calculated by
\begin{equation}\label{equ_Rd}
R_{D,v} = B_l^{sub}\mathbb{E}\left\{ {\log_2 \left( {1 + \gamma _{D,v} } \right)} \right\}.
\end{equation}
According to the uniform distribution of the D2D users, the average throughput of a D2D user can be defined as $R_D = R_{D,v}$.

\begin{proposition}\label{Rd}
The throughput of a D2D user over the licensed spectrum is given by
\begin{equation}\label{Rd_exp}
R_D = {B_l^{sub}} \vartheta \left(P_D, A_1 Q_1,  f_{d_{v,b}}(y), f_{d_{v,v}}(r) \right),
\end{equation}
where $A_1 =  \int_0^{2\pi }  {\int_0^{{r_{cell}}} {\frac{1}{{1 + s{P_C}\mathcal {H}^{ - \alpha /2}{{(x,\theta ,y)}}}}  \frac{{2x}}{{r_{cell}^2}} \frac{1}{{2\pi }}dx} }  d\theta$,
$Q_1 = \mathcal{Q}\left(\lambda'_D, P_D, y\right)$, $\lambda'_{D} = \frac{1}{K_l} \left( \lambda_{D} - \frac{1}{\pi r_{cell}^2}\right)$, ${f_{{d_{v,b}}}}(y) = \frac{{2y}}{{r_{cell}^2}},~0 \le y \le {r_{cell}}$.
$\vartheta(\cdot)$, $\mathcal{Q}(\cdot)$ and $\mathcal {H}(\cdot)$ are given in equations (\ref{function_eta}) to (\ref{function_H}).
\end{proposition}
\begin{proof}
See Appendix \ref{Proof2}.
\end{proof}

\subsection{Throughput Analysis in the Unlicensed Spectrum}
In the unlicensed spectrum, the existence of Wi-Fi users cannot be neglected. In order to share the unlicensed spectrum with the Wi-Fi system, both LTE-U and D2D-U users adopt the LBE-based LBT mechanism. As such, each user cannot access the unlicensed subchannel all the time. Therefore, we introduce the medium access probability~(MAP)~\cite{BLK-2017} to represent the probability for each user to transmit on its allocated unlicensed subchannel.

Based on the random back-off mechanism of LTE-U, D2D-U and Wi-Fi users, their waiting counters can be characterized by HCPPs. In addition, we assume that a channel/subchannel is sensed busy by user $u$ when the received power at user $u$ from other existing active users exceeds an energy detection threshold. Thresholds for LTE-U, D2D-U, and Wi-Fi users are denoted as $P^{th}_{CU}$, $P^{th}_{DU}$ and $P^{th}_{W}$, respectively. Besides, we define the sets of active LTE-U, D2D-U and Wi-Fi users as $\psi^{act}_{CU}$, $\psi^{act}_{DU}$, and $\psi^{act}_{W}$ with densities $\lambda^{act}_{CU}$, $\lambda^{act}_{DU}$, and $\lambda^{act}_{W}$, respectively. According to the results in \cite{BLK-2017}, the MAPs for LTE-U, D2D-U, and Wi-Fi users can be respectively given by
\begin{align}
&p_{CU}^{act} = \mathbb{P}\left( {u \in \psi^{act}_{CU}|u \in \psi_{CU}} \right) = \mathcal{M}(P^{th}_{CU}, \frac{\lambda_{CU}}{K_u}, \frac{\lambda_{DU}}{K_u}, \lambda_{W}),\label{MAP_CU}\\
&p_{DU}^{act} = \mathbb{P}\left( {v \in \psi^{act}_{CU}|v \in \psi_{DU}} \right) = \mathcal{M}(P^{th}_{DU}, \frac{\lambda_{CU}}{K_u}, \frac{\lambda_{DU}}{K_u}, \lambda_{W}),\label{MAP_DU}\\
&p_{W}^{act} = \mathbb{P}\left( {a \in \psi^{act}_{W}|a \in \psi_{W}} \right) = \mathcal{M}(P^{th}_{W}, \lambda_{CU}, \lambda_{DU}, \lambda_{W}),\label{MAP_W}
\end{align}
where
\begin{equation}\label{MAP_exp}
\begin{split}
&\mathcal{M}(P^{th}, \lambda_{1}, \lambda_{2}, \lambda_{3}) = 1 - \\
&~\frac{\alpha}{2\pi }\int_0^\infty {\int_0^1 {\frac{{\sin (z\mathcal {S})}}{{z{e^{z\mathcal {C}}}}}{e^{ - {P^{th}}{{\left( {\frac{{z\alpha \mathcal {S}/(2\varepsilon {\pi ^2})}}{{{\lambda _1}P_C^{2/\alpha } + {\lambda _2}P_D^{2/\alpha } + {\lambda _3}P_W^{2/\alpha }}}} \right)}^{\alpha /2}}}}} } d\varepsilon dz,
\end{split}
\end{equation}
with $\mathcal {S} = \sin\left(\frac{2\pi}{\alpha}\right)$ and $\mathcal {C} = \cos\left(\frac{2\pi}{\alpha}\right)$. In addition, according to~\cite{BLK-2017}, we can also have $\lambda^{act}_{CU} = p_{CU}^{act} \lambda_{CU}$, $\lambda^{act}_{DU} = p_{DU}^{act} \lambda_{DU}$, and $\lambda^{act}_{W} = p_{W}^{act} \lambda_{W}$. In what follows, we will analyze the throughput for LTE-U, D2D-U, and Wi-Fi users, respectively.

\subsubsection{LTE-U Users}
Since different LTE-U users utilize orthogonal subchannels, the interference to LTE-U user $u$ only comes from active D2D-U users in the same subchannel and active Wi-Fi APs. We can express the SINR of LTE-U user $u$ as
\begin{equation}
\gamma_{{CU},u} = \frac{{P_C}d_{u,b}^{ - \alpha }h}{{{\sigma ^2} + {I_{DU}} + {I_W}}},
\end{equation}
where ${I_{DU}} = \mathcal {I}(b,P_D,{{\psi'_{DU}} \cap \psi_{DU}^{act}})$ and ${I_{W}} = \mathcal {I}(b,P_W,\psi _{W}^{act})$ are the interference from D2D-U and Wi-Fi users, respectively. Then, we can calculate the throughput for LTE-U user $u$ as
\begin{equation}\label{equ_Rcu}
R_{CU,u} = p^{act}_{CU}B_u^{sub}\mathbb{E}\left\{ {\log_2 \left( {1 + \gamma_{{CU},u} } \right)} \right\},
\end{equation}
and the expectation of the throughput of an LTE-U can be defined as $R_{CU} = R_{CU,u}$.

\begin{proposition}\label{Rcu}
The throughput of an LTE-U user over the unlicensed spectrum is given by
\begin{equation}\label{Rcu_exp}
{R_{CU}} = p^{act}_{CU} B_u^{sub} \vartheta\left(P_C, Q_1 Q_2, \delta(y), {f_{{d_{u,b}}}}(r)\right),
\end{equation}
where $Q_1 = \mathcal{Q}\left(p^{act}_{DU}\lambda'_{DU}, P_D,0\right)$, $Q_2 = \mathcal{Q}\left(p^{act}_{W}\lambda_{W},P_W,0\right)$, $\lambda'_{DU} = \frac{\lambda_{DU}}{K_u}$ and $\delta(y)$ is the Dirac function. $\vartheta(\cdot)$, $\mathcal{Q}(\cdot)$ and $\mathcal {H}(\cdot)$ are given in equations (\ref{function_eta}) to (\ref{function_H}).
\end{proposition}
The detailed proof of Proposition~\ref{Rcu} is omitted, because $R_{CU}$ can be derived from equation~(\ref{equ_Rcu}) by the same methods in Appendices~\ref{Proof1} and~\ref{Proof2}.

\subsubsection{D2D-U Users}
Due to the orthogonal subchannel utilization of LTE-U users, D2D-U user $v$ may only be interfered with one LTE-U user in the same subchannel, denoted by $u$. The interference to D2D-U user $v$ also comes from other co-channel active D2D-U users and all active Wi-Fi users. Then, the SINR of D2D-U user $v$ can be expressed as
\begin{equation}
\gamma_{{DU},v} = \frac{{{P_D}d_{v,v}^{ - \alpha }h}}{{{\sigma ^2} + {I_{CU}} + {I_{DU}} + {I_W}}},
\end{equation}
where ${I_{CU}} = \mathcal {I}(v,P_C,{\left\{ u \right\} \cap \psi_{CU}^{act}})$, ${I_{DU}} = \mathcal {I}(v,P_D,\psi'_{DU} \cap \psi_{DU}^{act}\setminus v)$ and ${I_{W}} = \mathcal {I}(v,P_W,\psi _{W}^{act})$ are the interference from LTE-U, D2D-U, and Wi-Fi users, respectively. Thus, the throughput for D2D-U user $v$ can be given by
\begin{equation}\label{equ_Rdu}
R_{{DU},v} = p^{act}_{DU}B_u^{sub}\mathbb{E}\left\{ {\log_2 \left( {1 + \gamma _{{DU},v} } \right)} \right\},
\end{equation}
and the expectation of the throughput of a D2D-U can be defined as $R_{DU} = R_{{DU},v}$.

\begin{proposition}\label{Rdu}
The throughput of a D2D-U user over the unlicensed spectrum is given by
\begin{equation}\label{Rdu_exp}
{R_{DU}} = p^{act}_{DU} B_u^{sub} \vartheta\left(P_D, A_2 Q_1 Q_2, {f_{{d_{v,b}}}}(y), {f_{{d_{v,v}}}}(r)\right),
\end{equation}
where $A_2 = p^{act}_{CU}\int_0^{2\pi } {\int_0^{{r_{cell}}} {\frac{1}{{1 + s{P_C}\mathcal {H}^{ - \alpha /2}{{(x,\theta ,y)}}}}  \frac{{2x}}{{r_{cell}^2}}  \frac{1}{{2\pi }}dx} }  d\theta + ( 1 - p^{act}_{CU} ) $, $Q_{1} = \mathcal{Q}\left(p^{act}_{DU}\lambda'_{DU}, P_D,y\right)$,
$\lambda'_{DU} = \frac{\lambda_{DU}}{K_u} - \frac{1}{K_u\pi r_{cell}^2}$, ${f_{{d_{v,b}}}}(y) = \frac{{2y}}{{r_{cell}^2}},~0 \le y \le {r_{cell}}$, $Q_2 = \mathcal{Q}\left(p^{act}_{W}\lambda_{W},P_W,y\right)$. $\vartheta(\cdot)$, $\mathcal{Q}(\cdot)$ and $\mathcal {H}(\cdot)$ are given from (\ref{function_eta}) to (\ref{function_H}).
\end{proposition}

Since we can derive $R_{DU}$ from equation (\ref{equ_Rdu}) through the same method introduced in Appendices~\ref{Proof1} and~\ref{Proof2}, the detailed proof for Proposition~\ref{Rdu} is omitted.

\subsubsection{Wi-Fi Users}
As for Wi-Fi users, let us define the distance between Wi-Fi user $w$ and its corresponding AP $a$ as $d_{w,a}$, whose PDF can be given by ${f_{d_{w,a}}}(r) = {2r}/{L_w^2},~ 0 \le r \le {L_w}$. Different from LTE-U or D2D-U users, Wi-Fi users occupy the whole unlicensed band for data transmissions. Therefore, Wi-Fi user $w$ may receive the interference from all active LTE-U and D2D-U users as well as other active Wi-Fi APs. The SINR of Wi-Fi user $w$ can be given by
\begin{equation}
\gamma_{{W},w} = \frac{{{P_W}d_{w,a}^{ - \alpha }h}}{{{\sigma ^2} + {I_{CU}} + {I_{DU}} + {I_W}}},
\end{equation}
where ${I_{CU}} = \mathcal {I}(a,P_C,\psi_{CU}^{act})$, ${I_{DU}} = \mathcal {I}(a,P_D,\psi_{DU}^{act})$ and ${I_{W}} = \mathcal {I}(a,P_W,\psi _{W}^{act}\setminus a)$ are the interference from LTE-U, D2D-U and Wi-Fi users, respectively. Then, we can calculate the throughput for Wi-Fi user $w$ by
\begin{equation}\label{equ_Rw}
R_{W,w} = p^{act}_{W}B_u \mathbb{E}\left\{ {\log_2 \left( {1 + \gamma_{{W},w} } \right)} \right\}.
\end{equation}
Finally, the expectation of the throughput of a Wi-Fi user can be defined as $R_{W} = R_{W,w}$.

\begin{proposition}\label{Rw}
The throughput of a Wi-Fi user over the unlicensed spectrum is given by
\begin{equation}\label{Rw_exp}
{R_{W}} = p^{act}_{W} B_u \vartheta\left(P_W, Q_1 Q_2 Q_3, {f_{{d_{a,b}}}}(y), {f_{{d_{w,a}}}}(r)\right),
\end{equation}
where $Q_1 = \mathcal{Q}\left(p^{act}_{CU}\lambda_{CU}, P_C,y\right)$, $Q_2 = \mathcal{Q}\left(p^{act}_{DU}\lambda_{DU}, P_D, y\right)$, $Q_3 = \mathcal{Q}\left(p^{act}_{W}\lambda'_{W},P_W,y\right)$, $\lambda'_{W} = \lambda _{W} - \frac{1}{\pi r_{cell}^2}$, ${f_{{d_{a,b}}}}(y) = \frac{{2y}}{{r_{cell}^2}},~0 \le y \le {r_{cell}}$. $\vartheta(\cdot)$, $\mathcal{Q}(\cdot)$ and $\mathcal {H}(\cdot)$ are given in equations (\ref{function_eta}) to (\ref{function_H}).
\end{proposition}

The detailed proof for Proposition~\ref{Rw} is also omitted, since the derivation of $R_W$ is similar to $R_{CU}$ and $R_{DU}$.

\section{Problem Formulation and Algorithm Design}
\label{Problem Formulation and Algorithm Design}

In this section, we aim to maximize the throughput of the system by optimizing the densities of four kinds of users. Since the problem is a global nonlinear optimization problem, which is NP-hard and difficult to solve, we design a SQP-based algorithm to obtain the suboptimal solutions of this problem.

\subsection{Problem Formulation}
In our system, each LTE/D2D user can transmit on either licensed or unlicensed spectrum for communications. Therefore, our objective is to maximize the throughput of the system by optimizing $\bm{\lambda} = [\lambda_C, \lambda_D, \lambda_{CU}, \lambda_{DU}]^T$. The throughput of the system is defined as the total throughput of all users over both licensed and unlicensed bands, which is given by $R_{total}(\bm{\lambda}) = S_{cell}\left({\lambda _C}{R_C} + {\lambda _D}{R_D} + {\lambda _{CU}}{R_{CU}} + {\lambda _{DU}}{R_{DU}} \right)$. Here, $S_{cell} = \pi r_{cell}^2$ is the area of the cell. We assume that the densities of all LTE and D2D users in the system are $\lambda_C^{all}$ and $\lambda_D^{all}$, respectively. Besides, to guarantee the QoS for different users, we also define the throughput requirements for LTE/LTE-U, D2D/D2D-U, and Wi-Fi users as ${R_{C}^{th}}$, ${R_{D}^{th}}$, and ${R_{W}^{th}}$, respectively. Then, the optimization problem can be formulated as follow:
\begin{subequations}\label{P1}
\begin{align}
\textbf{P1:}~~\mathop {\max }\limits_{\bm{\lambda}}&~~ R_{total}(\bm{\lambda}) \label{P1_Objective} \\
s.t.&~~{\lambda_C} + {\lambda_{CU}} \leq \lambda_C^{all}, \label{P1_constraint_ccu}\\
&~~{\lambda_D} + {\lambda_{DU}} \leq \lambda_D^{all}, \label{P1_constraint_ddu}\\
&~~{R_{C}} \ge {R_{C}^{th}},\quad{R_{CU}} \ge {R_{C}^{th}},\label{P1_constraint_Rccu}\\
&~~{R_{D}} \ge {R_{D}^{th}},\quad{R_{DU}} \ge {R_{D}^{th}},\label{P1_constraint_Rddu}\\
&~~{R_{W}} \ge {R^{th}_{W}},\label{P1_constraint_Rw}\\
&~~{\lambda_C}\geq 0,~{\lambda_D}\geq 0,~{\lambda_{CU}}\geq 0,~{\lambda_{DU}}\geq 0.\label{P1_constraint nn}
\end{align}
\end{subequations}
Constraint~(\ref{P1_constraint_ccu}) and (\ref{P1_constraint_ddu}) are the density requirements for LTE/LTE-U and D2D/D2D-U users, respectively. Constraints (\ref{P1_constraint_Rccu})-(\ref{P1_constraint_Rw}) are the QoS requirements for LTE/LTE-U, D2D/D2D-U, and Wi-Fi users, respectively. Constraint~(\ref{P1_constraint nn}) indicates the non-negativity of the densities.

\subsection{Spectrum Access Algorithm Design}
In this part, we propose a SQP-based spectrum access algorithm to obtain the suboptimal solutions of problem~(P1). We first transform the optimization problem into quadratic programming~(QP) problems. Then, a SQP-based method is presented to solve these QP problems.

\subsubsection{Problem Transformation}
Note that problem~(P1) is a non-convex optimization problem with nonlinear constraints, it is difficult to find the optimal solution. Motivated by the SQP technique, which is efficient to tackle with certain classes of large scale nonlinear programming problems~\cite{BD,FSTQ-2016}, we first transform problem~(P1) into a series of QP problems at feasible points, and then solve these problems to obtain the suboptimal solution iteratively.

More specifically, we first rewrite the optimization problem in the following form:
\begin{equation}\label{USP_standard}
\begin{split}
\mathop {\min } \limits_{ \bm{\lambda} }  ~~&  f(\bm{\lambda})\\
s.t. ~~& g_i(\bm{\lambda}) \leq 0,\quad(i = 1,...,N)
\end{split}
\end{equation}
where $\bm{\lambda}$ is the vector of optimization variables and $N$ is the number of constraints. Then, given a feasible iteration point $\bm{\lambda}^k$, we can approximate the objective function and constraints into quadratic terms of Taylor series at $\bm{\lambda}^k$, i.e.,
\begin{equation}\label{SQP_Taylor}
\begin{split}
\mathop {\min } \limits_{ \emph{\textbf{S}} } ~~&  f(\bm{\lambda}^k) + [\nabla f(\bm{\lambda}^k)]^T \emph{\textbf{S}} + \frac{1}{2}\emph{\textbf{S}} ^T \emph{\textbf{H}}^k \emph{\textbf{S}} \\
s.t. ~~& g_i(\bm{\lambda}^k) + [\nabla g_i(\bm{\lambda}^k)]^T \emph{\textbf{S}} \leq 0, \quad (i = 1,...,N)
\end{split}
\end{equation}
where $\emph{\textbf{H}}^k = \nabla^2 f( \bm{\lambda}^k ) $ denotes an approximation of the Hessian matrix of the objective function at feasible point $\bm{\lambda}^k$ and $\emph{\textbf{S}} = \bm{\lambda} - \bm{\lambda}^k$ denotes the search direction~\cite{CYQXZ}. We define problem~(\ref{SQP_Taylor}) as the QP problem at point $\bm{\lambda}^k$, since its objective function and constraints are expressed in quadratic polynomials.

\subsubsection{SQP-based Algorithm}
After transforming problem~(P1) into QP problems, we introduce a SQP-based method to obtain their solutions.
Denoting the optimal solution of problem~(\ref{SQP_Taylor}) as $\emph{\textbf{S}}^*$, we can calculate the next iteration point $\bm{\lambda}^{k+1}$ by
\begin{equation}\label{line_search}
\bm{\lambda}^{k+1} = \bm{\lambda}^k + \beta^k \emph{\textbf{S}}^*
\end{equation}
where $\beta^k$ is the step length which can be calculated by the line search procedure to achieve a sufficient decrease in objective function for convergence~\cite{CYQXZ}. In addition, the matrix $\emph{\textbf{H}}$ in equation~(\ref{SQP_Taylor}) can be updated by the Broyden-Fletcher-Goldfarb-Shanno~(BFGS) quasi-Newton method \cite{LV} below
\begin{equation}\label{BFGS}
\emph{\textbf{H}}^{k+1} = \emph{\textbf{H}}^{k} + \frac{ \Delta q^{k+1} [\Delta q^{k+1}]^T }{ [\Delta q^{k+1}]^T [\Delta \bm{\lambda}^{k+1}] } - \frac{\emph{\textbf{H}}^{k} \Delta \bm{\lambda}^{k+1} [\Delta \bm{\lambda}^{k+1}]^T \emph{\textbf{H}}^{k}}{ [\Delta \bm{\lambda}^{k+1}]^T \emph{\textbf{H}}^{k} \Delta \bm{\lambda}^{k+1} }
\end{equation}
where $\Delta \bm{\lambda}^{k+1} = \bm{\lambda}^{{k+1}}-\bm{\lambda}^{k}$ and $\Delta q^{k+1} = \nabla f(\bm{\lambda}^{{k+1}}) - \nabla f(\bm{\lambda}^{k})$.
In each iteration, we update $\bm{\lambda}$ and $\emph{\textbf{H}}$ according to equations~(\ref{line_search}) and~(\ref{BFGS}) until
\begin{equation}\label{terminal_condition}
|{f(\bm{\lambda}^{k+1}) - f(\bm{\lambda}^{k})}|< \tau
\end{equation}
where $\tau$ denotes the error tolerance threshold. The suboptimal solution of problem~(\ref{USP_standard}) is then obtained as $\bm{\lambda}^* = \bm{\lambda}^{k+1}$.
The procedure of this SQP-based method can be summarized in Algorithm~\ref{SQP_Algorithm}.

\begin{algorithm}[!t]
\caption{SQP-based method.}\label{SQP_Algorithm}
\KwIn{Initial value $\bm{\lambda}^0 $. }
\KwOut{Suboptimal solution $\bm{\lambda}^* $. }
\Begin
{
Set $k = 0,~\emph{\textbf{H}}^0 = \textbf{\emph{I}}~$(identity matrix);\\
\Repeat{$\bm{\lambda}^k$ satisfies the terminal condition~(\ref{terminal_condition})}
{
Simplify $f(\bm{\lambda})$ and $g_i(\bm{\lambda}), (i = 1,...,N)$ as quadratic terms of Taylor series at $\bm{\lambda}^k$;\\
Find the optimal search direction $\emph{\textbf{S}}^*$;\\
Update $\bm{\lambda}$ according to equation (\ref{line_search});\\
Update $\emph{\textbf{H}}$ according to equation (\ref{BFGS});\\
$k := k + 1$;\\
}
Terminate with the final output $\bm{\lambda}^* $.\\
}
\end{algorithm}

Since the performance of Algorithm~\ref{SQP_Algorithm} is greatly affected by its initial values, it is important to set an appropriate value for $\bm{\lambda}^{0}$. Therefore, instead of random or all-zero initialization, we adopt the greedy policy to find approximate solutions of problem~(P1) and then set them as initial value $\bm{\lambda}^{0}$, which can effectively accelerate the convergence of our proposed algorithm. According to the greedy policy, each user sequentially accesses to the spectrum which achieves the higher total throughput while satisfying all constraints. When the spectrum access of a user may decrease the total throughput in the system, this user cannot access the corresponding spectrum.

\section{Performance Analysis}
\label{Performance Analysis}

In this section, we first present the analysis on the convergence and the complexity of our proposed spectrum access algorithm. Then, we characterize the throughput regions for the system to analyze its performance and illustrate the spectrum access issue.

\subsection{Algorithm Analysis}
\subsubsection{Convergence}
According to equation~(\ref{line_search}), the optimal solution of the QP problem at $\bm{\lambda}^k$ is given by $\bm{\lambda}^{k+1}$. Thus, according to the property of QP, we have $f(\bm{\lambda}^{k+1}) \leq f(\bm{\lambda}^k)$, implying that the value of the objective function is non-increasing after each iteration. Besides, since the optimal value of problem~(P1) is $f(\bm{\lambda}^*)$, we also have $f(\bm{\lambda}^k) \geq f(\bm{\lambda}^*)$, $\forall k$, which indicates that the value of the objective function is lower bounded by a finite value. Therefore, after finite iterations, the terminal condition~(\ref{terminal_condition}) will be satisfied, which ensures the convergence of our proposed algorithm.

\subsubsection{Complexity}
Since we adopt the greedy policy to obtain initial value $\bm{\lambda}^{0}$, the computational complexity of initialization is given by $O(2M_{C}+2M_{D})$, in which $M_C = \pi r_{cell}^2 \lambda^{all}_C$ and $M_D = \pi r_{cell}^2 \lambda^{all}_D$ denote the number of all LTE and D2D users in the system. When the system contains the same number of LTE and D2D users, e.g., $M_C = M_D = M$, the computational complexity of initialization is $O(4M) = O(M)$. Besides, a QP problem should be solved in each iteration of the SQP algorithm. As referred in~\cite{S_SG}, the computational complexity for each QP problem is $O(V^3)$, where $V = 2$ corresponds to the number of variables.

\subsection{Throughput Region Analysis}
In this part, we characterize the throughput regions for LTE/LTE-U and D2D/D2D-U users over licensed and unlicensed spectrum, respectively. Assume the total density of LTE and LTE-U users in the system is no higher than $\lambda_C^{all}$. Besides, the densities of D2D and D2D-U users in the system is given by $\lambda_D$ and $\lambda_{DU}$. When the density of LTE users is $\lambda_C$, we can calculate the total throughput of LTE users as $\widetilde{R}_C = S_{cell}\lambda_{C} R_C$. Then, the total throughput of LTE-U users, whose density is $\lambda_{CU}$, can be calculated by $\widetilde{R}_{CU} = S_{cell}\lambda_{CU} R_{CU}$. Therefore, we can depict the throughput region for LTE/LTE-U users as
\begin{equation}
\Omega_{C}(\lambda_C^{all}) =  \left\{ (\widetilde{R}_C,\widetilde{R}_{CU})\in \mathbb{R}^2_{+}|\lambda_{C} + \lambda_{CU} \leq \lambda_C^{all} \right\}.
\end{equation}
Likewise, we assume that the maximum total density of D2D and D2D-U users in the system is $\lambda_D^{all}$. As such, given $\lambda_C$ and $\lambda_{CU}$, we can calculate the total throughput for D2D and D2D-U users, with densities $\lambda_D$ and $\lambda_{DU}$, as $\widetilde{R}_D = S_{cell}\lambda_{D} R_D$ and $\widetilde{R}_{DU} = S_{cell}\lambda_{DU} R_{DU}$, respectively. Therefore, the throughput region for D2D/D2D-U users is given by
\begin{equation}
\Omega_{D}(\lambda_D^{all}) =  \left\{ (\widetilde{R}_D,\widetilde{R}_{DU})\in \mathbb{R}^2_{+}|\lambda_{D} + \lambda_{DU} \leq \lambda_D^{all} \right\}.
\end{equation}
Due to the non-negativity of the throughput, the boundaries of $\Omega_{C}$ and $\Omega_{D}$ depict curves in the first quadrant on $\mathbb{R}^2$. It is also worthwhile to mention that either $\Omega_{C}$ or $\Omega_{D}$ has the complete information on the users densities in the system. In other words, for any point on the boundary of $\Omega_{C}$, we can find its counterpart which denotes the same user densities in the boundary of $\Omega_{D}$, and vice versa.

Furthermore, we can characterize the throughput region of the system to illustrate the spectrum access problem. Assume the system includes LTE/D2D users over the licensed band with the density of $\lambda_l = \lambda_C + \lambda_D$, and LTE-U/D2D-U users over the unlicensed band with the density of $\lambda_u = \lambda_{CU} + \lambda_{DU}$, respectively. Besides, we introduce $\kappa_l$ and $\kappa_u$ to denote the proportions of LTE and LTE-U users over licensed and unlicensed spectrum, which are given by $\kappa_l = {\lambda_C}/{\lambda_l}$ and $\kappa_u ={\lambda_{CU}}/{\lambda_u}$, respectively. Therefore, we have $\lambda_C = \kappa_l\lambda_l$, $\lambda_D = (1-\kappa_l)\lambda_l$, $\lambda_{CU} = \kappa_u\lambda_u$, and $\lambda_{DU} = (1-\kappa_u)\lambda_u$. Based on these notations, we can calculate the total throughput over licensed and unlicensed spectrum by $\widetilde{R}_l = S_{cell}(\lambda_{C} R_C + \lambda_{D} R_D)$ and $\widetilde{R}_u = S_{cell}(\lambda_{CU} R_{CU} + \lambda_{DU} R_{DU})$, respectively. Therefore, given the densities of all LTE and D2D users, i.e., $\lambda^{all}_C$ and $\lambda^{all}_D$, the throughput region of the system can be depicted as
\begin{equation}
\begin{split}
\Omega(\lambda) &= \left\{ (\widetilde{R}_l,\widetilde{R}_u)\in \mathbb{R}^2_{+}  |\kappa_l\lambda_l + \kappa_u\lambda_u \leq \lambda^{all}_C, \right. \\ & \qquad\qquad\left. (1-\kappa_l)\lambda_l + (1-\kappa_u)\lambda_u \leq \lambda^{all}_D \right\}.
\end{split}
\end{equation}
Similar to $\Omega_C$ and $\Omega_D$, the boundary of $\Omega$ depicts a curve in the first quadrant on $\mathbb{R}^2$. In addition, for any $(\widetilde{R}_l,\widetilde{R}_u)\in\Omega$, the total throughput of the system, denoted by $\widetilde{R}$, can be calculated by $\widetilde{R} = \widetilde{R}_l + \widetilde{R}_u$. Thus, we can investigate the spectrum access issue of the system by maximizing $\widetilde{R}$, which is given by the following problem:
\begin{equation}
\textbf{P3:}~~\mathop {\max }\limits_{(\widetilde{R}_l,\widetilde{R}_u) \in \Omega}~\widetilde{R},
\end{equation}
whose optimal solution is denoted by $(\widetilde{R}_l^*,\widetilde{R}_u^*)$. Since problem~(P3) is a standard linear programming problem, the optimal solution $(\widetilde{R}_l^*,\widetilde{R}_u^*)$ is obtained when the line $\widetilde{R} = \widetilde{R}_l + \widetilde{R}_u$ is tangent to the throughput region $\Omega$ on $\mathbb{R}^2$. As such, the user densities corresponding to the point $(\widetilde{R}_l^*,\widetilde{R}_u^*)$, denoted by~($\lambda_C^*, \lambda_D^*, \lambda_{CU}^*, \lambda_{DU}^*$), is the optimal solution in the system.

\section{Simulation Results}
\label{Simulation Result}

In this section, we present the simulation results of the geometry analysis, our proposed spectrum access algorithm and the throughput regions, respectively. The simulation parameters based on the existing LTE-Advanced specifications \cite{EUTRA} are given in Table \ref{parameters}.

\begin{table}[!t]
\centering
\caption{Parameters for Simulation} \label{parameters}
\vspace{-3mm}
\begin{tabular}{|p{2.4in}|p{0.5in}|}
  \hline \textbf{Parameters} & \textbf{Values}\\
  \hline
  \hline
  Radius of the cell $r_{cell}$ & 200 m\\
  Communication radius of D2D/D2D-U users $L_d$ & 20 m\\
  Communication radius of Wi-Fi users $L_w$ & 25 m\\
  Transmit power of LTE/LTE-U user $P_C$ & 17 dBm\\
  Transmit power of D2D/D2D-U user $P_D$ & 10 dBm\\
  Transmit power of Wi-Fi user $P_W$ & 23 dBm\\
  Licensed bandwidth $B_l$ & 40 MHz\\
  Unlicensed bandwidth $B_u$ & 500 MHz\\
  Decay factor of the path loss $\alpha$ & 4 \\
  Energy detection threshold for LTE-U user $P^{th}_{CU}$ & -62 dBm\\
  Energy detection threshold for D2D-U user $P^{th}_{DU}$ & -62 dBm\\
  Energy detection threshold for Wi-Fi user $P^{th}_{W}$ & -62 dBm\\
  Throughput threshold for cellular/cellular-U user $R^{th}_{C}$ & 100 Mbps \\
  Throughput threshold for D2D/D2D-U user $R^{th}_{D}$ & 100 Mbps \\
  Throughput threshold for Wi-Fi user $R^{th}_{W}$ & 54 Mbps \\
  Error tolerance threshold $\tau$ & 0.01 \\
\hline
\end{tabular}

\end{table}

\subsection{Verification of Geometry Analysis}
To verify our geometry analysis, we provide the simulated and analytical results on the throughput of different users, which are presented in Fig.~\ref{simulation_Rcd} to Fig.~\ref{simulation_Rw}.

\begin{figure}[!t]
\centering
\vspace{-5mm}
\includegraphics[width=3.2in]{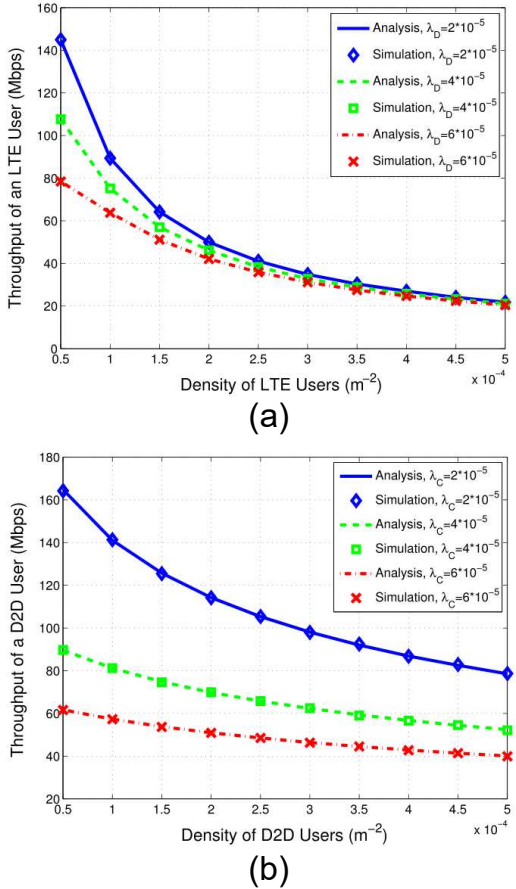}
\vspace{-7mm}
\caption{(a)~Throughput of an LTE user vs. density of LTE users with different densities of D2D users; (b)~Throughput of a D2D user vs. density of D2D users with different densities of LTE users.}
\vspace{-5mm}
\label{simulation_Rcd}
\end{figure}

For LTE and D2D users, their throughput is mainly related to the bandwidth of subchannels and the co-channel interference, as shown in equations~(\ref{Rc_exp}) and~(\ref{Rd_exp}). Fig.~\ref{simulation_Rcd}(a) shows the throughput of LTE users versus their density with different densities of D2D users. The throughput of an LTE user $R_C$ monotonically decreases with the density of LTE users $\lambda_C$. As mentioned in Section~\ref{System Model}, the number of licensed subchannels is proportional to $\lambda_C$, and thus, as $\lambda_C$ increases, the bandwidth occupied by each LTE user decreases, leading to the decrease of $R_C$. Besides, $R_{C}$ decreases as the density of D2D users $\lambda_D$ grows, since more D2D links will interfere an LTE user. In Fig.~\ref{simulation_Rcd}(b), we present the throughput of D2D users versus their density. Likewise, the throughput of a D2D user $R_{D}$ monotonically decreases with an increase of $\lambda_D$ as well as $\lambda_C$. The reason lies in that the increase of $\lambda_D$ brings severer interference from other D2D links, and the increase of $\lambda_C$ reduces the bandwidth occupied by each D2D link. By comparing these two figures, we can infer that the $R_C$ is comparable to $R_D$ in the system with a low user densities~(e.g., both $\lambda_C$ and $\lambda_D$ are below $6 \times 10^{-5} m^{-2}$).

\begin{figure}[!t]
\centering
\vspace{-4mm}
\includegraphics[width=3.1in]{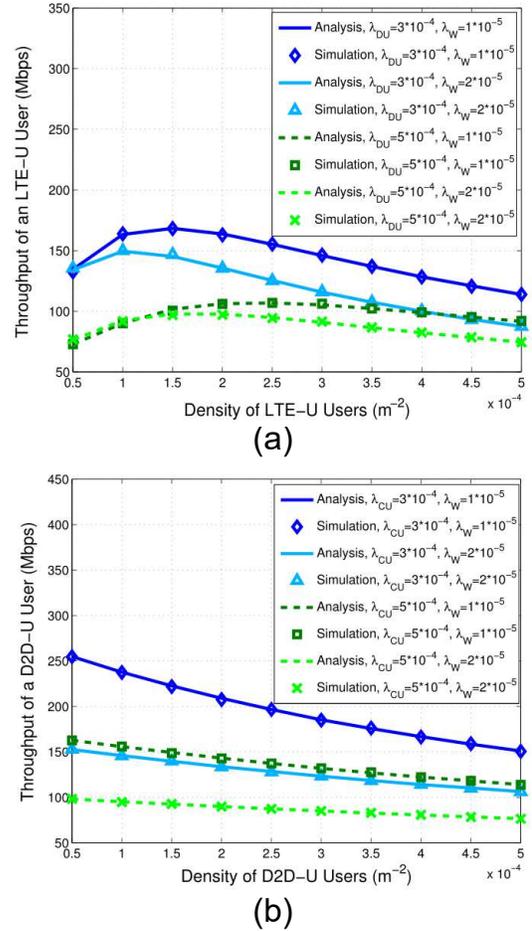}
\vspace{-6mm}
\caption{(a)~Throughput of an LTE-U user vs. density of LTE-U users with different densities of D2D-U users and Wi-Fi APs; (b)~Throughput of a D2D-U user vs. density of D2D-U users with different densities of LTE-U users and Wi-Fi APs.}
\vspace{-3mm}
\label{simulation_Rcudu}
\end{figure}

The throughput of LTE-U, D2D-U and Wi-Fi users is mainly affected by their MAPs, i.e., $p^{act}_{CU}$, $p^{act}_{DU}$ and $p^{act}_{W}$. In Fig.~\ref{simulation_Rcudu}(a), we plot the throughput of LTE-U users versus their density with different densities of D2D-U users and Wi-Fi APs. The throughput of an LTE-U user $R_{CU}$ first increases and then decreases with the density of LTE-U users $\lambda_{CU}$. According to equation~(\ref{MAP_CU}), $\lambda_{CU}$ increases with $p^{act}_{CU}$, leading to the increase of $R_{CU}$. However, based on equation~(\ref{MAP_DU}), as $\lambda_{CU}$ grows, $p^{act}_{DU}$ increases, which decreases $R_{CU}$. Therefore, there exists a trade-off between $p^{act}_{CU}$ and $p^{act}_{DU}$. When $\lambda_{CU}$ is low, the throughput is dominated by the increase of $p^{act}_{CU}$, leading to the growth of $R_{CU}$, but when $\lambda_{CU}$ reaches to a high level, the throughput is dominated by the increase of $p^{act}_{DU}$, which decreases $R_{CU}$. In addition, $R_{CU}$ decreases with the density of D2D-U users $\lambda_{DU}$ as well as the density of Wi-Fi APs $\lambda_W$, since the increase of $\lambda_{DU}$ and $\lambda_W$ results in the decrease of $p^{act}_{CU}$, which decreases $R_{CU}$.
Fig.~\ref{simulation_Rcudu}(b) presents the throughput of D2D-U users versus their density. Similar to Fig.~\ref{simulation_Rcd}(b), the throughput of a D2D-U user $R_{DU}$ monotonically decreases with $\lambda_{DU}$, since a higher $\lambda_{DU}$ implies a lower $p^{act}_{DU}$, which leads to the decrease of $R_{DU}$. Besides, $R_{DU}$ decreases with the increase of $\lambda_{CU}$ and $\lambda_{W}$, due to the decrease of $p^{act}_{DU}$.
Combing these two figures, we can find out that $R_{DU}$ is higher than $R_{CU}$ in the low traffic system~(where both $\lambda_{CU}$ and $\lambda_{DU}$ are lower than $3 \times 10 ^{-4} m^{-2}$). Moreover, since the bandwidth of the unlicensed spectrum~(500MHz) is much wider than that of the licensed spectrum~(40MHz), LTE-U/D2D-U users are more likely to achieve a higher throughput than LTE/D2D over the licensed spectrum.

\begin{figure}[!t]
\centering
\includegraphics[width=3in]{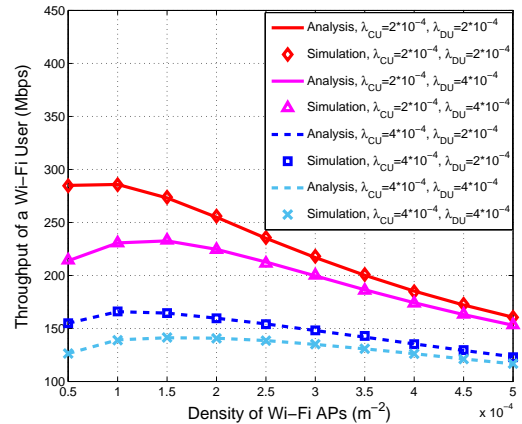}
\vspace{-3mm}
\caption{Throughput of a Wi-Fi user vs. density of Wi-Fi APs with different densities of LTE-U and D2D-U users.}
\vspace{-5mm}
\label{simulation_Rw}
\end{figure}

Fig.~\ref{simulation_Rw} presents the throughput of Wi-Fi users versus their density with different densities of LTE-U and D2D-U users. With the increase of $\lambda_W$, the throughput of a Wi-Fi user $R_W$ first increases and then decreases. According to equations~(\ref{MAP_CU}) to~(\ref{MAP_W}), when $\lambda_W$ increases, both $p^{act}_{CU}$ and $p^{act}_{DU}$ decrease, and thus, $R_W$ increases. However, as $\lambda_W$ grows, $p^{act}_{W}$ decreases, leading to the decrease of $R_W$. Similar to LTE-U users, there exist the trade-off between the aforementioned two effects for $R_W$. When $\lambda_W$ is low, $R_W$ is dominated by the decrease of $p^{act}_{CU}$ and $p^{act}_{DU}$, i.e., $R_W$ increases with $\lambda_W$. When $\lambda_W$ reaches a high level, $R_W$ is mainly affected by the decrease of $p^{act}_{W}$, and thus, $R_W$ decreases with $\lambda_W$. In addition, $R_W$ decreases with $\lambda_{CU}$ and $\lambda_{DU}$ due to the decrease of $p^{act}_{W}$.

From Fig.~\ref{simulation_Rcd} to Fig.~\ref{simulation_Rw}, we can easily find out that the analytical results are consistent with the simulated ones, which validates our analysis.

\subsection{Performance of Spectrum Access Algorithm}
The performance of our proposed spectrum access algorithm is presented in Fig.~\ref{spectum_access_performance} to Fig.~\ref{spectrum_access_WiFi}.

\begin{figure}[!t]
\centering
\vspace{-2mm}
\includegraphics[width=3.1in]{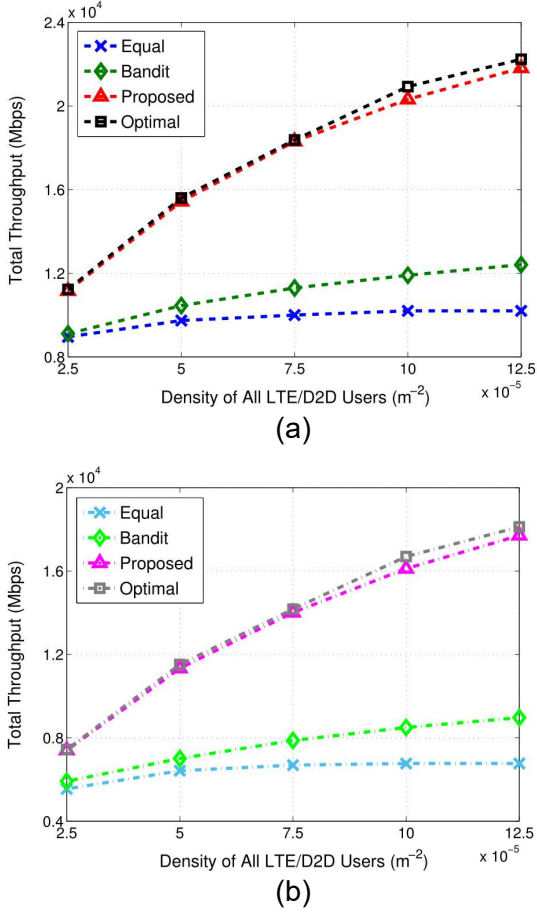}
\vspace{-7mm}
\caption{Performance comparison with different densities of Wi-Fi APs $\lambda_W$;~(a)~$\lambda_W = 1\times10^{-5}m^{-2}$;~(b)~$\lambda_W = 2\times10^{-5}m^{-2}$.}
\vspace{-5mm}
\label{spectum_access_performance}
\end{figure}

In Fig.~\ref{spectum_access_performance}, we compare our proposed algorithm with other two algorithms:
\begin{itemize}
  \item \textbf{Equal Proportion}: the densities of different users are identical.
  \item \textbf{Bandit Algorithm}: each user sequentially access the spectrum based on their utilities, as referred in~\cite{SD-2017}.
\end{itemize}
We also present the simulated optimal values obtained by traversal in this figure. For the sake of fairness, we assume the densities of all LTE and D2D users are identical, i.e., $\lambda_C^{all} = \lambda_D^{all}$. As shown in Fig.~\ref{spectum_access_performance}, the total throughput monotonically increases with the densities of all LTE/D2D users. By comparing the performance of these three algorithms, we find out that our proposed spectrum access algorithm outperforms the others. Besides, the performance of our proposed algorithm is quite close to the optimal value, which indicates the effectiveness of our proposed algorithm. For all algorithms, the increase of $\lambda_W$ will lead to the decline of the total throughput, since more interference is brought by the Wi-Fi users over the unlicensed spectrum.

\begin{figure*}[!t]
\centering
\includegraphics[width=6.4in]{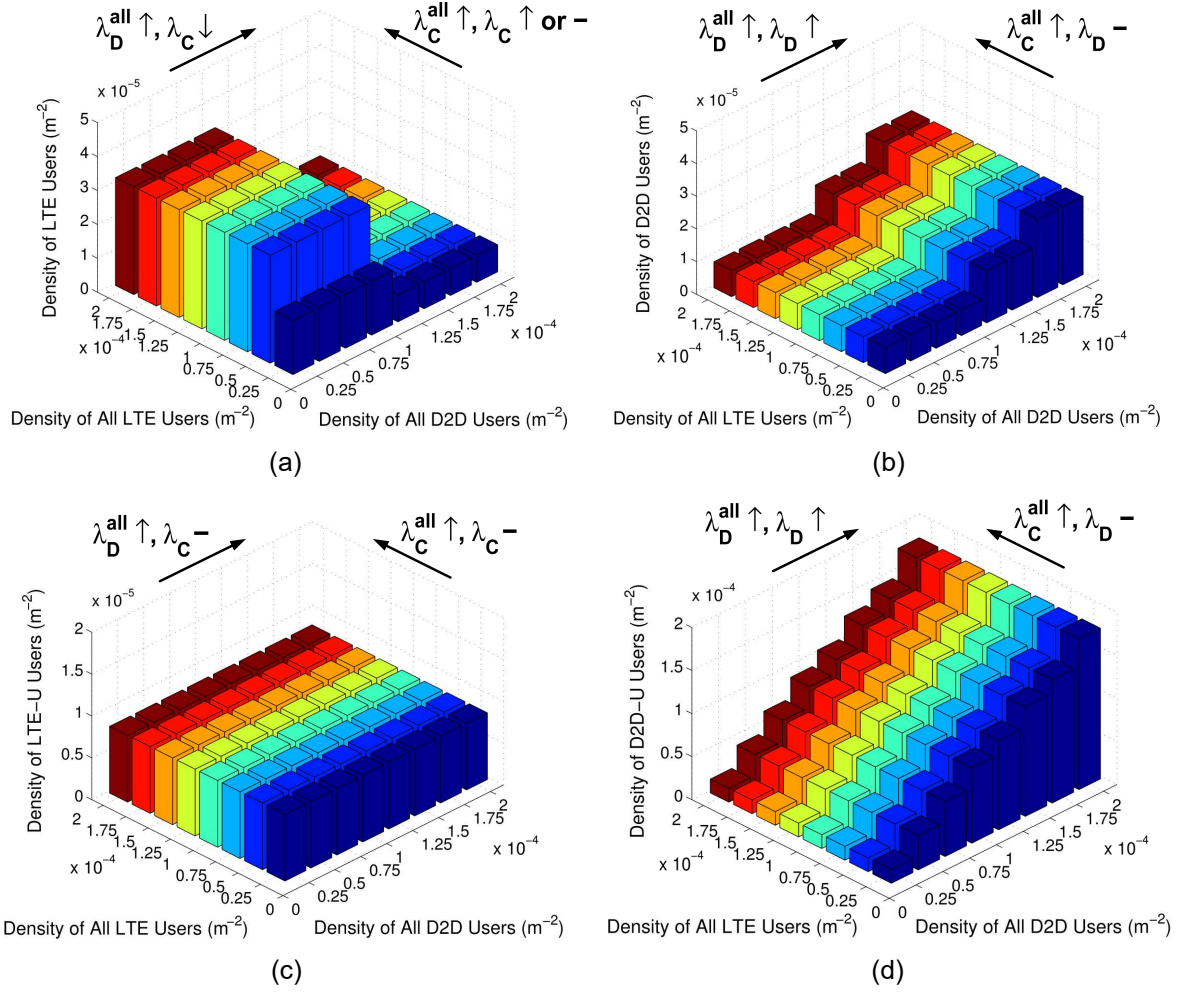}
\vspace{-5mm}
\caption{Spectrum access with different densities of all LTE/D2D users, given the density of Wi-Fi APs $\lambda_W = 1 \times 10^{-4} m^{-2}$; (a)~LTE users;~(b)~D2D users;~(c)~LTE-U users;~(d)~D2D-U users. The notations $\uparrow$, $\downarrow$, and $-$ indicate increasing, decreasing, and remaining unchanged, respectively.}
\vspace{-5mm}
\label{surf}
\end{figure*}

Fig.~\ref{surf} shows the spectrum access with different densities of all LTE and D2D users, i.e., $\lambda_C^{all}$ and $\lambda_D^{all}$, given the fixed density of Wi-Fi APs $\lambda_W = 1\times10^{-4} m^{-2}$. We can observe that the total user density over the unlicensed spectrum is higher than that of the licensed one. This is because that the bandwidth of the unlicensed spectrum is much wider than that of the licensed spectrum, and thus, utilizing the unlicensed spectrum can achieve a higher throughput. Besides, when $\lambda_D^{all}$ increases, the densities of LTE/LTE-U users decrease or remain unchanged, since the interference to the LTE/LTE-U users becomes severer as more D2D/D2D-U users access the spectrum. However, the spectrum access of D2D/D2D-U users is hardly affected by the increase of $\lambda_C^{all}$. This is mainly because that the transmitter and the receiver of a D2D pair are close to each other, which can guarantee a satisfactory SINR even the densities of LTE/LTE-U are large.

\begin{figure}[!t]
\centering
\includegraphics[width=3in]{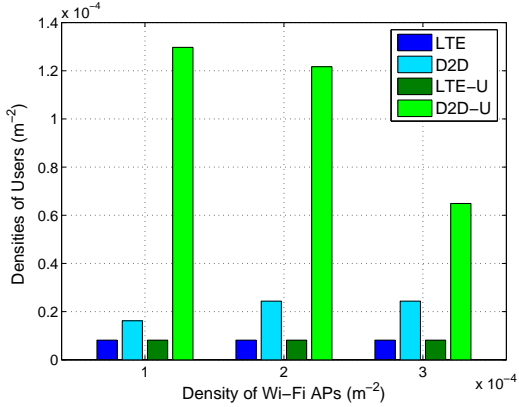}
\vspace{-3mm}
\caption{Spectrum access with different densities of Wi-Fi APs $\lambda_W$, given the densities of all LTE/D2D users $\lambda_C^{all} = \lambda_D^{all} = 1.5 \times 10^{-4} m^{-2}$.}
\vspace{-5mm}
\label{spectrum_access_WiFi}
\end{figure}

In Fig.~\ref{spectrum_access_WiFi}, we present the result of spectrum access with different densities of Wi-Fi APs given the densities of users on the LTE/D2D modes $\lambda_C^{all} = \lambda_D^{all} = 1.5 \times 10^{-4} m^{-2}$. As $\lambda_W$ increases, the density of D2D-U users decreases while the density of D2D users increases. This is because that more Wi-Fi APs can bring severer interference over the unlicensed spectrum, and thus, more D2D pairs prefer the licensed spectrum than the unlicensed one. Besides, when $\lambda_W$ grows, the densities of LTE/LTE-U users remain unchanged, which is due to the QoS requirements of LTE/D2D users over the licensed spectrum. In all cases, we can observed that the densities of D2D/D2D-U users are larger than that of LTE/LTE-U users, since the proximity of the transmitter and the receiver of a D2D pair can ensure a higher SINR.

\subsection{Throughput Regions}
The throughput regions of the system are plotted in Fig.~\ref{region_LTE_D2D} and Fig.~\ref{region_optimal}.

\begin{figure}[!t]
\centering
\includegraphics[width=3.2in]{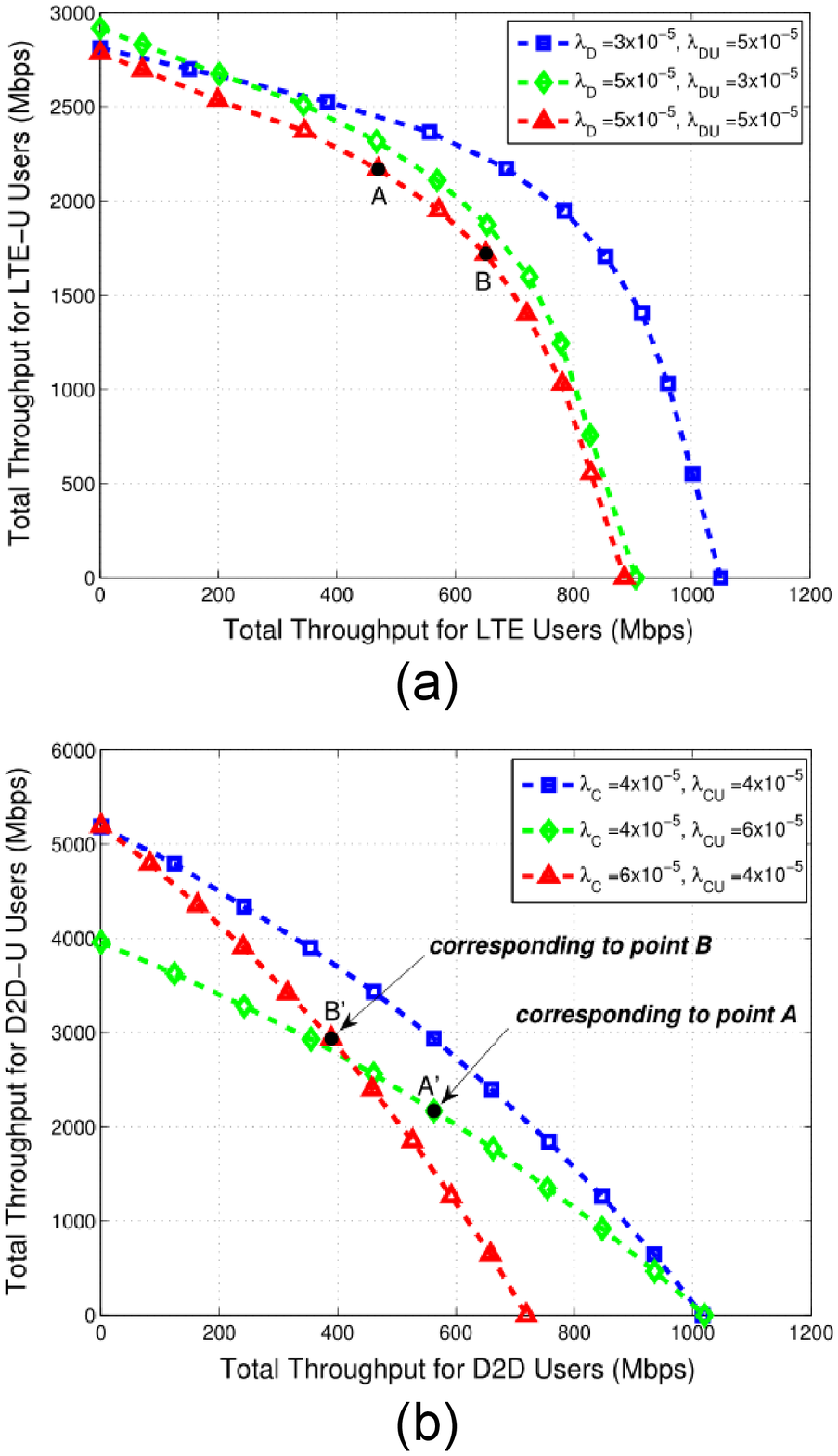}
\vspace{-7mm}
\caption{(a)~Throughput regions of LTE and LTE-U users over both licensed and unlicensed spectrum, given the density of all LTE users $\lambda_C^{all} = 1 \times 10^{-4} m^{-2}$ and the density of Wi-Fi users $\lambda_W = 3 \times 10^{-5} m^{-2}$; (b)~Throughput regions of D2D and D2D-U users over both licensed and unlicensed spectrum, given the density of all D2D users $\lambda_D^{all} = 1 \times 10^{-4} m^{-2}$ and the density of Wi-Fi APs $\lambda_W = 3 \times 10^{-5} m^{-2}$.}
\label{region_LTE_D2D}
\vspace{-5mm}
\end{figure}

In Fig.~\ref{region_LTE_D2D}(a), we depict the throughput regions of LTE/LTE-U users over both licensed and unlicensed spectrum with different densities of D2D/D2D-U users, given $\lambda_C^{all} = 1 \times 10^{-4} m^{-2}$ and $\lambda_W = 3 \times 10^{-5} m^{-2}$. It can be seen that as the increase of $\lambda_D$ and $\lambda_{DU}$, the throughput regions shrink, due to the increase of interference links over the licensed and unlicensed spectrum. Besides, in Fig.~\ref{region_LTE_D2D}(b), we present the throughput regions of D2D/D2D-U users over both licensed and unlicensed spectrum with different densities of LTE/LTE-U users, given $\lambda_D^{all} = 1 \times 10^{-4} m^{-2}$ and $\lambda_W = 3 \times 10^{-5} m^{-2}$. Similarly, the increase of $\lambda_C$ and $\lambda_{CU}$ shrinks the throughput regions, since the interference links over licensed and unlicensed spectrum increase. As discussed in the Section~\ref{Performance Analysis}, either of these two figures has the complete information on the densities of different users in the system, i.e., for any point in either figure, we can find its counterpart with the same user densities in the other figure. For instance, point $A$~(where $\lambda_{C} = 4 \times 10^{-5} m^{-2}$, $\lambda_{CU} = 6 \times 10^{-5} m^{-2}$) on the curve with $\lambda_{D} = 5 \times 10^{-5} m^{-2}$, $\lambda_{DU} = 5 \times 10^{-5} m^{-2}$ in Fig.~\ref{region_LTE_D2D}(a) is corresponding to point $A'$~(where $\lambda_{D} = 5 \times 10^{-5} m^{-2}$, $\lambda_{DU} = 5 \times 10^{-5} m^{-2}$) on the curve with $\lambda_{C} = 4 \times 10^{-5} m^{-2}$, $\lambda_{CU} = 6 \times 10^{-5} m^{-2}$ in Fig.~\ref{region_LTE_D2D}(b). As point $A$ moves to point $B$~(where $\lambda_{C} = 6 \times 10^{-5} m^{-2}$, $\lambda_{CU} = 4 \times 10^{-5} m^{-2}$) along with the same curve in Fig.~\ref{region_LTE_D2D}(a), the counterpart of point $B$, namely $B'$~(where $\lambda_{D} = 5 \times 10^{-5} m^{-2}$, $\lambda_{DU} = 5 \times 10^{-5} m^{-2}$), moves to the curve with $\lambda_{C} = 6 \times 10^{-5} m^{-2}$, $\lambda_{CU} = 4 \times 10^{-5} m^{-2}$ in Fig.~\ref{region_LTE_D2D}(b).

Fig.~\ref{region_optimal}(a) shows the throughput regions of the system over both licensed and unlicensed spectrum with different proportions of LTE/LTE-U users $\kappa_l$ and $\kappa_u$, given $\lambda^{all}_C = \lambda^{all}_D = 5\times 10^{-4} m^{-2}$ and $\lambda_W = 8 \times 10^{-5} m^{-2}$. With the increase of $\kappa_l$, the total throughput over the licensed spectrum increases, i.e., throughput regions extend. Due to the orthogonal utilization of the LTE mode, the number of subchannels increases with the density of LTE users, and thus, the co-channel interference of D2D users is alleviated effectively, thereby increasing the total throughput over the licensed spectrum. Besides, the total throughput over the unlicensed spectrum decreases with $\kappa_u$, since the bandwidth of subchannels decreases with the density of LTE-U users. Moreover, we can easily obtain the optimal densities in the system under different $\kappa_l$ and $\kappa_u$ by depicting the tangent points of lines $\widetilde{R} = \widetilde{R}_l + \widetilde{R}_u$ and throughput regions, which denoted by points $A$, $B$, $C$ and $D$ in Fig.~\ref{region_optimal}(a). For instance, when $\kappa_l = 0.8$ and $\kappa_u = 0.7$, the optimal densities in the system are achieved at point $A$, whose densities are given by $(\lambda_C^*, \lambda_D^*, \lambda_{CU}^*, \lambda_{DU}^*) = (5.71\times10^{-5}m^{-2}, 1.43\times10^{-5}m^{-2}, 4.50\times 10^{-4}m^{-2}, 1.93\times 10^{-4}m^{-2})$.

\begin{figure}[!t]
\centering
\vspace{-1mm}
\includegraphics[width=3.2in]{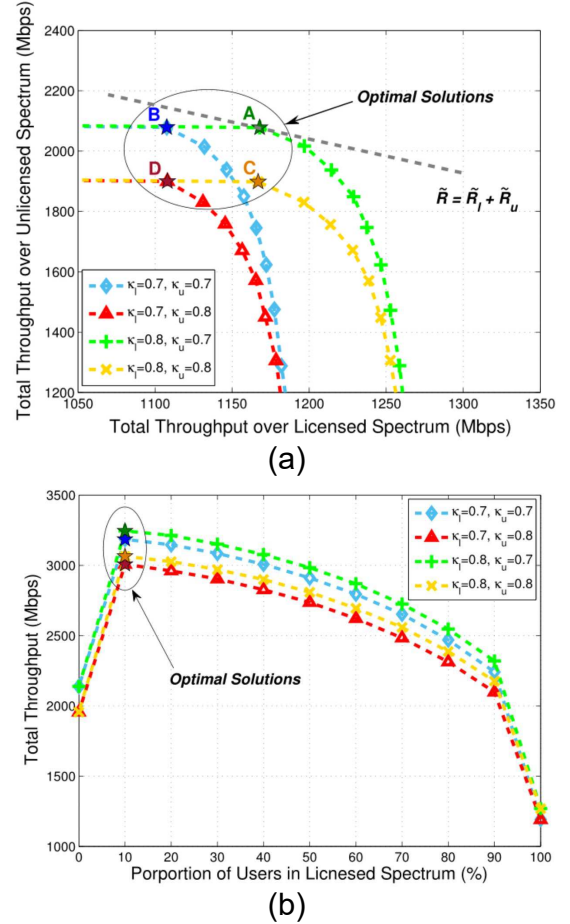}
\vspace{-5mm}
\caption{(a)~Throughput regions of the system with different $\kappa_l$ and $\kappa_u$, given the densities of all LTE and D2D users $\lambda^{all}_C = \lambda^{all}_D = 5\times 10^{-4} m^{-2}$, and the density of Wi-Fi APs $\lambda_W = 8 \times 10^{-5} m^{-2}$; (b)~Total throughput of the system with different $\kappa_l$ and $\kappa_u$, given the densities of all LTE and D2D users $\lambda^{all}_C = \lambda^{all}_D = 5\times 10^{-4} m^{-2}$, and the density of Wi-Fi APs $\lambda_W = 8 \times 10^{-5} m^{-2}$.}
\vspace{-5mm}
\label{region_optimal}
\end{figure}

In Fig.~\ref{region_optimal}(b), we present the total throughput of the system with respect to the proportions of users over the licensed spectrum, defined by $\rho = \frac{\lambda_l}{\lambda_l+\lambda_u}$, with different $\kappa_l$ and $\kappa_u$, given $\lambda^{all}_C = \lambda^{all}_D = 5\times 10^{-4} m^{-2}$ and $\lambda_W = 8 \times 10^{-5} m^{-2}$. The total throughput of the system first increases and then decreases with $\rho$. In addition, when $\rho$ is low~(e.g., less than $10\%$), more LTE and D2D users will access the licensed spectrum. Therefore, the interference over the unlicensed spectrum is effectively alleviated, and the total throughput of the system increases. When $\rho$ reaches to a specific level, the total throughput of the system starts to decrease~(e.g., the transition approximately occurs at $10\%$ when $\kappa_l = 0.8$ and $\kappa_u = 0.7$), since the throughput of LTE/D2D users over the licensed spectrum is less than that of LTE-U/D2D-U users over the unlicensed spectrum. Moreover, we depict the maxima of the total throughput under different $\kappa_l$ and $\kappa_u$ as points $A$, $B$, $C$ and $D$, respectively, which are corresponding to their counterparts in Fig.~\ref{region_optimal}(a). We can observe that the optimal proportions of users over the unlicensed spectrum~($90\%$) are much larger than those of the licensed spectrum~($10\%$) for different $\kappa_l$ and $\kappa_u$. In other words, users are more likely to utilize the unlicensed spectrum rather than the licensed one.

\subsection{Remarks on Spectrum Access}

Based on the above simulation results, we present the following remarks on the spectrum access of LTE/D2D users in the system.

\begin{remark}
Given the fixed density of Wi-Fi APs, more users are likely to access the unlicensed spectrum rather than the licensed one in the low-traffic scenarios~(i.e., where the densities of all LTE/D2D users are low), since the bandwidth of unlicensed spectrum is much richer than that of the licensed spectrum. As the network traffic increases~(i.e., the densities of all LTE/D2D users increase), the unlicensed spectrum gets crowded, and thus, new users tend to utilize the licensed spectrum to avoid the severe interference in the unlicensed spectrum. However, when the network traffic is very high, new users may access neither the licensed nor the unlicensed spectrum. The spectrum access under different network traffics is summarized in Table.~\ref{Access Table1}.
\end{remark}

\begin{table}[!t]
\centering
\caption{Spectrum Access under Different Network Traffics}
\vspace{-3mm}
\begin{tabular}{|p{1.01in}<{\centering}|p{0.89in}<{\centering}|p{0.89in}<{\centering}|}
  \hline
  \diagbox{\tabincell{c}{Network\\Traffic}}{Spectrum} & Licensed & Unlicensed \\
  \hline
  Low & \tabincell{c}{less likely to\\access} & \tabincell{c}{more likely to\\access} \\
  \hline
  High & \tabincell{c}{more likely to\\access for new users} & fully utilized \\
  \hline
\end{tabular}
\label{Access Table1}
\end{table}

\begin{remark}
Given the fixed densities of all LTE/D2D users, when the system contains few Wi-Fi APs, more users are willing to access the unlicensed spectrum rather than licensed one for higher throughput, since the interference from the Wi-Fi system is trivial. However, in a network with high density of Wi-Fi APs, users prefer the licensed spectrum to avoid the interference from the Wi-Fi system. The spectrum access under different densities of Wi-Fi APs is summarized in Table.~\ref{Access Table2}.
\end{remark}

\begin{table}[!t]
\centering
\caption{Spectrum Access under Different Wi-Fi Densities}
\vspace{-3mm}
\begin{tabular}{|p{1.01in}<{\centering}|p{0.89in}<{\centering}|p{0.89in}<{\centering}|}
  \hline
  \diagbox{\tabincell{c}{Wi-Fi\\Density}}{Spectrum} & Licensed & Unlicensed \\
  \hline
  Low & \tabincell{c}{less likely to\\access} & \tabincell{c}{more likely to\\access} \\
  \hline
  High & \tabincell{c}{more likely to\\access} & \tabincell{c}{less likely to\\access} \\
  \hline
\end{tabular}
\label{Access Table2}
\end{table}

\begin{remark}
Given the fixed densities of all LTE/D2D users and Wi-Fi APs, D2D users are more willing to access the unlicensed spectrum than LTE users. This is because that the transmitter and the receiver of each D2D pair are usually close to each other, which can ensure satisfactory SINR even in some high-traffic scenarios with dense interference links. The unlicensed spectrum access of LTE and D2D users is summarized in Table.~\ref{Access Table3}.
\end{remark}

\begin{table}[!t]
\centering
\caption{Unlicensed Spectrum Access of LTE and D2D Users}
\vspace{-3mm}
\begin{tabular}{|p{1.01in}<{\centering}|p{0.89in}<{\centering}|p{0.89in}<{\centering}|}
  \hline
  User Type & LTE & D2D \\
  \hline
  \tabincell{c}{Unlicensed Spectrum\\Access} & \tabincell{c}{less willing to\\access} & \tabincell{c}{more willing to\\access} \\
  \hline
\end{tabular}
\label{Access Table3}
\end{table}

\section{Conclusion}
\label{Conclusion}

In this paper, we have presented the geometric analysis for a unified network where D2D users work as an underlay to LTE users in both licensed and unlicensed spectrum. We have analytically derived the throughput for each kind of users by leveraging stochastic geometry based on proposed spectrum sharing schemes. To investigate the spectrum access problem in the system, we have maximized the total throughput of the system by optimizing the densities of different kinds of users, and then proposed a SQP-based spectrum access algorithm to obtain the suboptimal solutions. Moreover, we have characterized throughput regions to illustrate the performance of the system.

Three conclusions can be drawn from the simulation results. First, users tend to access the unlicensed spectrum in a low-traffic network, while new users prefer the licensed spectrum as the network traffic increases. Second, given the fixed network traffic, as the number of Wi-Fi APs increases, more users are prone to access the licensed spectrum. Third, given the fixed network traffic, D2D users are more willing to access the unlicensed spectrum than LTE users.

\begin{appendices}
\section{Proof of Proposition 1}\label{Proof1}
According to equation (\ref{equ_Rc}), we have:
\begin{align}
R_{C,u} &= B_l^{sub}\mathbb{E}\left\{ {\log_2 \left( {1 + \gamma _{C,u} } \right)} \right\} \nonumber\\
&= B_l^{sub}\mathbb{E}_{I_D,d_{u,b}} \left\{ \int_0^\infty  \mathbb{P} [ {\log_2 ( {1 + \frac{{{P_C}d_{u,b}^{ - \alpha }h}}{{{\sigma ^2} + {I_D}}}} ) > t} ]dt \right\} \nonumber\\
&= B_l^{sub}\mathbb{E}_{I_D,d_{u,b}} \left\{ \int_0^\infty  {\mathbb{P}[ {h > \frac{{{2^t} - 1}}{{{P_C}d_{u,b}^{ - \alpha }}}( {{\sigma ^2} + {I_D}} )} ]} dt\right\} \nonumber\\
&\mathop  = \limits^{(a)} B_l^{sub}\mathbb{E}_{I_D,d_{u,b}} \left\{ {\int_0^\infty  {e^{ { - \frac{{{2^t} - 1}}{{{P_C}d_{u,b}^{ - \alpha }}}\left( {{\sigma ^2} + {I_D}} \right)} }} dt} \right\} \nonumber\\
&\mathop  = \limits^{(b)} B_l^{sub} \mathbb{E}_{d_{u,b}} \left\{ \int_0^\infty \frac{e^{ -s{\sigma ^2}} \mathbb{E}_{I_D} \left\{ {{e^{ - s{I_D}}}} \right\}}{{\left( {s + {P_C^{ - 1}}{d_{u,b}^\alpha }} \right)\ln 2}} ds \right\}\nonumber\\
&\mathop  = \limits^{(c)} B_l^{sub} \int_0^{{r_{cell}}} {\int_0^\infty  \frac{e^{ -s{\sigma ^2}}  \mathbb{E}\left\{ {e^{-sI_D}}\right\}}{{\left( {s + {P_C^{ - 1}}{r^\alpha }} \right)\ln 2}}  \frac{{2r}}{{r_{cell}^2}}dsdr},
\end{align}
where (a) is given by $h \sim \exp(1)$, (b) follows the substitution $s = \frac{{{2^t} - 1}}{{{P_C}d_{u,b}^{ - \alpha }}}$ and (c) is based on the PDF of $d_{u,b}$. The term $\mathbb{E} \left\{ {e^{ - s I_D}}\right\}$ in (c) can be calculated by the probability generating functional (PGFL) of the PPP as
\begin{align}
\mathbb{E}\left\{ {e^{-sI_D}}\right\} &= \mathbb{E}_{h,d_{v,b}}\left\{ {e^ { - s \sum_{{v} \in {\psi '_D}} {{P_D}d_{v,b}^{ - \alpha }h} } } \right\} \nonumber \\
&= \mathbb{E}_{h,d_{v,b}} \left\{ {\prod\limits_{v \in {\psi '_D}} {e^ { { - s{P_D}d_{{v},b}^{ - \alpha }h} }}} \right\} \nonumber \\
&= \mathbb{E}_{d_{v,b}} \left\{ {\prod\limits_{v \in {\psi '_D}} {\frac{1}{{1{ + }s{P_D}d_{{v},b}^{ - \alpha }}}} } \right\} \nonumber \\
&= e^{ { - \frac{\lambda_D}{K_l}\int_0^{2\pi } {\int_0^{{r_{cell}}} {\frac{1}{{1 + {{\left( {s{P_D}} \right)}^{ - 1}}{x^{\alpha} }}}x} } dxd\theta } },
\end{align}
Then we can get the throughput of an LTE user and the proof ends.

\section{Proof of Proposition 2}\label{Proof2}
According to equation (\ref{equ_Rd}), we have:
\begin{align}
R_{D,v} &= B^{sub}_l\mathbb{E}\left\{ {{{\log }_2}\left( {1 + \gamma_{D,v}} \right)} \right\}\nonumber\\
 &= B^{sub}_l  \int_0^{{L_d}} {\int_0^\infty \frac{e^{ -s{\sigma ^2}}  \mathbb{E}\left\{ {{e^{ - s{I_C}}}} \right\} \mathbb{E}\left\{ {{e^{ - s{I_D}}}} \right\} }{{\left( {s + {P_D^{ - 1}}{r^\alpha }} \right)\ln 2}}   \frac{{2r}}{L_d^2}dsdr}.
\end{align}
The term $\mathbb{E} \left\{ {e^{ - s I_C}}\right\}$ can be given by
\begin{align}
\mathbb{E}\left\{ {{e^{ - s{I_C}}}} \right\} &= \mathbb{E}_{h,d_{u,v}}\left\{ {{e^{ - s{P_C}d_{u,v}^{ - \alpha }h}}} \right\} \nonumber\\
&= \mathbb{E}_{d_{u,v}}\left\{ {\frac{1}{{1+s{P_C}d_{u,v}^{ - \alpha }}}} \right\}.
\end{align}
When $r_{cell} \gg L_d$, the distance from LTE user $u$ to the RX of D2D user $v$ can be approximately replaced by the distance from LTE user $u$ to the TX of D2D user $v$. Denote the distance from LTE user $u$ and the TX of D2D user $v$ to the BS (the center of the cell) as $d_{u,b}$ and $d_{v,b}$, respectively. Due to the character of PPP, the PDF of $d_{u,b}$ is given by ${f_{{d_{u,b}}}}(r) = {2r}/{r_{cell}^2},~0 \le r \le {r_{cell}}$. Then, we can calculate $\mathbb{E} \left\{ {e^{ - s I_C}}\right\}$ as
\begin{equation}
\begin{split}
\mathbb {E} \left\{ {{e^{ - s{I_C}}}} \right\} = \int_0^{2\pi } \int_0^{{r_{cell}}} \frac{1}{{1 + s{P_C}\mathcal {H}{{(x,\theta ,d_{v,b})}^{ - \alpha /2}}}} \cdot \\ \quad\frac{{2x}}{{r_{cell}^2}} \frac{1}{{2\pi }}dx  d\theta.
\end{split}
\end{equation}
Besides, the term $\mathbb{E} \left\{ {e^{ - s I_D}}\right\}$ can be given by
\begin{align}
\mathbb{E}\left\{ {{e^{ - s{I_D}}}} \right\} &= \mathbb{E}_{h,d_{v',v}}\left\{ {e^{ { - s\sum_{{v'} \in {\psi '_D}\backslash v} {{P_D}d_{{v'},v}^{ - \alpha }h} } }} \right\} \nonumber\\
&\mathop \approx \limits^{(a)}  e^ { - \lambda'_D \int_0^{2\pi } {\int_0^{{r_{cell}}} {\frac{1}{{1 + {{\left( {s{P_D}} \right)}^{ - 1}}{{\cal H}^{\alpha /2}}\left( {x,\theta ,d_{v,b}} \right)}}x} } dxd\theta },
\end{align}
where in~(a) $\lambda'_D = \frac{1}{K_l}\left(\lambda_D -\frac{1}{\pi r_{cell}^2}\right)$, and the PDF of $d_{v',v}$ can be approximately replaced by the PDF of distance between the TXs of D2D user $v'$ and $v$, given $r_{cell} \gg L_d$.
According to the character of PPP, the PDF of $d_{v,b}$ is given by ${f_{{d_{v,b}}}}(r) = {2r}/{r_{cell}^2},~0 \le r \le {r_{cell}}$. Then we can get the throughput of a D2D user and ends the proof.

\end{appendices}

\bibliographystyle{IEEEtran}

\begin{thebibliography}{40}

\bibitem{FHBJL-2019}
F.~Wu, H.~Zhang, B.~Di, J.~Wu, and L.~Song, ``Network controlled D2D communications: Licensed or unlicensed spectrum?'' in \emph{Proc. IEEE ICC}, Shanghai, China, May~2019.

\bibitem{WP_Rising}
``Rising to meet the 1000x mobile data challenge,'' White Paper, Qualcomm, San Jose, CA, USA, 2012.
[Online]. Available: http://www.qualcomm.com/media/documents/rising-meet-1000x-mobile-data-challenge.

\bibitem{AQV-2014}
A.~Asadi, Q.~Wang and V.~Mancuso, ``A survey on device-to-device communication in cellular networks,'' \emph{IEEE Commun. Surveys Tuts.}, vol.~16, no.~4, pp.~1801-1819, fourth quarter~2014.

\bibitem{KMCCK-2009}
K.~Doppler, M.~Rinne, C.~Wijting, C.B.~Ribeiro and K.~Hugl, ``Device-to-device communication as an underlay to LTE-Advanced networks,'' \emph{IEEE Commun. Mag.}, vol.~47, no.~12, pp.~42-49, Dec.~2009.

\bibitem{GEGSNGZ-2012}
G.~Fodor, E.~Dahlman, G.~Mildh, S.~Parkvall, N.~Reider, G.~Miklos and Z.~Turanyi, ``Design aspects of network assisted device-to-device communications,'' \emph{IEEE Commun. Mag.}, vol.~50, no.~3, pp.~170-177, Mar.~2012.

\bibitem{HLZ-2016}
H.~Zhang, L.~Song and Z.~Han, ``Radio resource allocation for device-to-device underlay communication using hypergraph theory", \emph{IEEE Trans. Wireless Commun.}, vol.~15, no.~7, pp.~4852-4861, Jul.~2016.

\bibitem{RYGNC-2016}
R.~Mahapatra, Y.~Nijsure, G.~Kaddoum, N.~U.~Hassan, and C.~Yuen, ``Energy efficiency
tradeoff mechanism towards wireless green communication: A survey,'' \emph{IEEE Commun. Surveys
Tuts.}, vol.~18, no.~1, pp.~686-705, first quarter 2016.

\bibitem{JJCWW-2016}
J.~Zuo, J.~Zhang, C.~Yuen, W.~Jiang, and W.~Luo, ``Energy efficient user association for
cloud radio access networks,'' \emph{IEEE Access}, vol.~4, pp.~2429-2438, May~2016.

\bibitem{HYL-2017}
H.~Zhang, Y.~Liao and L.~Song, ``D2D-U: Device-to-device communications in unlicensed bands for 5G system", \emph{IEEE Trans. Wireless Commun.}, vol.~16, no.~6, pp.~3507-3519, Jun.~2017.

\bibitem{WP_Extending}
``Extending LTE advanced to unlicensed spectrum,'' White Paper, Qualcomm, San Jose, CA, USA, 2013.
[Online]. Available: https://www.slideshare.net/qualcommwirelessevolution/extending-lte-advanced-to-unlicensed-spectrum-31732634

\bibitem{YWHLSXJ-2016}
Y.~Wu, W.~Guo, H.~Yuan, L.~Li, S.~Wang, X.~Chu and J.~Zhang, ``Device-to-device meets LTE-unlicensed,'' \emph{IEEE Commun. Mag.}, vol.~54, no.~5, pp.~154-159, May~2016.

\bibitem{RMLZXL-2015}
R.~Zhang, M.~Wang, L.~Cai, Z.~Zheng, X.~Shen and L.~Xie, ``LTE-unlicensed: The future of spectrum aggregation for cellular networks,'' \emph{IEEE Wireless Commun.}, vol.~22, no.~3, pp.~150-159, Jul.~2015.

\bibitem{BJYJ-2017}
B.~Chen, J.~Chen, Y.~Gao and J.~Zhang, ``Coexistence of LTE-LAA and Wi-Fi on 5 GHz with corresponding deployment scenarios: A survey,'' \emph{IEEE Commun. Surv. Tut.}, vol.~19, no.~1, pp.~7-32, first quarter~2017.

\bibitem{CHL-2015}
C. Yao, H. Zhang and L. Song, ``Demo: WiFi multihop: Implementing device-to-device local area networks by android smartphones," in \emph{Proc.~ACM Mobicom}, Hangzhou, China, Jun.~2015.

\bibitem{YSRYC-2013}
Y.~Liu, S.~Xie, R.~Yu, Y.~Zhang, and C.~Yuen, ``An efficient MAC protocol with selective grouping and cooperative sensing in cognitive radio networks,'' \emph{IEEE Trans. Veh. Technol.}, vol.~62, no.~8, pp.~3928-3941, Oct.~2013.

\bibitem{Netmanias}
Netmanias, ``Analysis of LTE-WiFi aggregation solutions,'' Mar.~2016.

\bibitem{RGFZ-2016}
R.~Liu, G.~Yu, F.~Qu, and Z.~Zhang, ``Device-to-device communications in unlicensed spectrum: Mode selection and resource allocation,'' \emph{IEEE Access}, vol.~4, pp.~4720-4729, Aug.~2016.

\bibitem{BMDWD-2017}
B.~Ismaiel, M.~Abolhasany, D.~Smithz, W.~Nix and D.~Franklinz, ``A survey and comparison of device-to-device architecture using LTE unlicensed band,'' in \emph{Proc. IEEE VTC}, Sydney, Australia, Jun.~2017.

\bibitem{HWS-2016}
H.~Yuan, W.~Guo and S.~Wang, ``Device-to-device communications in LTE-unlicensed heterogeneous network,'' in \emph{Proc. IEEE SPAWC}, Edinburgh, UK, Jul.~2016.

\bibitem{RRYS-2017}
R.~Bajracharya, R.~Shrestha, Y.~B.~Zikria and S.~W.~Kim, ``LTE or LAA: Choosing network mode for my mobile phone in 5G network,'' in Proc. \emph{IEEE VTC Spring}, Sydney, NSW, Australia, Jun.~2017.

\bibitem{FEEMR-2015}
F.~Liu, E.~Bala, E.~Erkip, M.~C.~Beluri and R.~ Yang, ``Small-cell traffic balancing over licensed and unlicensed bands,'' \emph{IEEE Trans. Veh. Technol.}, vol.~64, no.~12, pp.~5850-5865, Dec.~2015.

\bibitem{MBPF-2016}
M.~Banagar, B.~Maham, P.~Popovski, and F.~Pantisano, ``Power distribution of device-to-device communications in underlaid cellular networks,'' \emph{IEEE Wireless Commun. Lett.}, vol.~5, no.~2, pp.~204-207, Apr.~2016.

\bibitem{PBHKL-2018}
P.~Wang, B.~Di, H.~Zhang, K.~Bian and L.~Song, ``Cellular V2X communications in unlicensed spectrum: Harmonious coexistence with VANET in 5G systems,'' \emph{IEEE Trans. Wireless Commun.}, vol.~17, no.~8, pp.~5212-5224, Aug.~2018.

\bibitem{SD-2017}
S.~Maghsudi, and D.~Niyato, ``On transmission mode selection in D2D-enhanced small cell networks,'' \emph{IEEE Wireless Commun. Lett.}, vol.~6, no.~5, pp.~618-621, Oct.~2017.

\bibitem{EC-1999}
E.~M.~Royer, and C.~Toh, ``A review of current routing protocols for ad hoc mobile wireless networks,'' \emph{IEEE Pers. Commun.}, vol.~6, no.~2, pp.~46-55, Apr.~1999.

\bibitem{SLAAUS}
\emph{3GPP TR 36.889, Study on licensed-assisted access to unlicensed spectrum~(Release 13)}, Jun.~2015.

\bibitem{RJH-2015}
R.~Tang, J.~Zhao, and H.~Qu, ``Joint optimization of channel allocation, link assignment and power control for device-to-device communication underlaying cellular network,'' \emph{China Commun.}, vol.~12, no.~12, pp.~92-100, Dec.~2015.

\bibitem{MMAWSM-2013}
M.~Bennis, M.~Simsek, A.~Czylwik, W.~Saad, S.~Valentin and M.~Debbah, ``When cellular meets WiFi in wireless small cell networks?,'' \emph{IEEE Commun. Mag.}, vol.~51, no.~6, pp.~44-50, Jun. 2013

\bibitem{KWY-2014}
K.~R.~Malekshan, W.~Zhuang and Y.~Lostanlen, ``An energy efficient MAC protocol for fully connected wireless ad hoc networks,'' \emph{IEEE Trans. Wireless Commun.}, vol.~13, no.~10, pp.~5729-5740, Oct.~2017.

\bibitem{BLK-2017}
B.~Shang, L.~Zhao and K.~Chen, ``Enabling device-to-device communications in LTE-unlicensed spectrum,'' in \emph{Proc. IEEE ICC}, Paris, France, May~2017.

\bibitem{BD}
B.~Wang and D.~Pu, ``An SQP algorithm with flexible penalty functions,'' in \emph{Proc. IEEE BIFE}, Hangzhou, China, Nov.~2013.

\bibitem{FSTQ-2016}
F.~Bai, S.~Huang, T.~Vidal-Calleja, and Q.~Zhang, ``Incremental SQP method for constrained optimization formulation in SLAM,'' in \emph{Proc. IEEE ICARCV}, Phuket, Thailand, Nov.~2016.

\bibitem{CYQXZ}
C.~Yang, Y.~Yu, Q.~Li, X.~Dong and Z.~Ren, ``An image denoising method based on nonsubsampled contourlet transform with SQP optimization,'' \emph{IEEE CCC}, Dalian, China, Jul.~2017.

\bibitem{LV}
L.~S.~Coelho and V.~C.~Mariani, ``Combining of chaotic differential evolution and quadratic programming for economic dispatch optimization with valve-point effect,'' \emph{IEEE Trans. Power Syst.}, vol.~21, no.~2, pp.~989-996, May~2006.

\bibitem{RHES}
R.~Irene, H.~Ramon, E.~Charles and S.~C.~Carlos, ``Quadratic programming feature selection,'' \emph{Journal of Machine Learning Research}, vol.~11. pp.~1491-1516, Apr.~2010.

\bibitem{S_SG}
S.~Boyd, ``Subgradient methods,'' [Online]. Available: http://www.stanford.edu/class/ee364b/lectures/subgrad\underline~method\underline~notes.pdf.

\bibitem{EUTRA}
\emph{3GPP TS 36.213, Evolved universal terrestrial radio access (EUTRA) physical layer procedures~(Release 12)}, Sep.~2014.

\end{thebibliography}

\begin{IEEEbiography}
[{\includegraphics[width=1in,height=1.25in,clip,keepaspectratio]{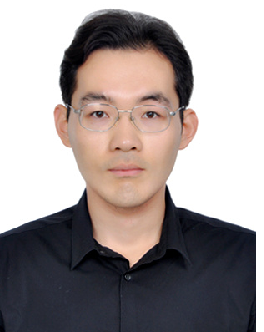}}]{Fanyi Wu} (S'18) received his B.~S.~degree in Electronic Engineering from Beihang University, Beijing, China, in 2016. He is currently pursuing his Ph.D. degree at School of Electrical Engineering and Computer Science in Peking University. His current research interest includes device-to-device communications, cellular Internet of unmanned aerial vehicles, and deep reinforcement learning.
\end{IEEEbiography}

\begin{IEEEbiography}
[{\includegraphics[width=1in,height=1.25in,clip,keepaspectratio]{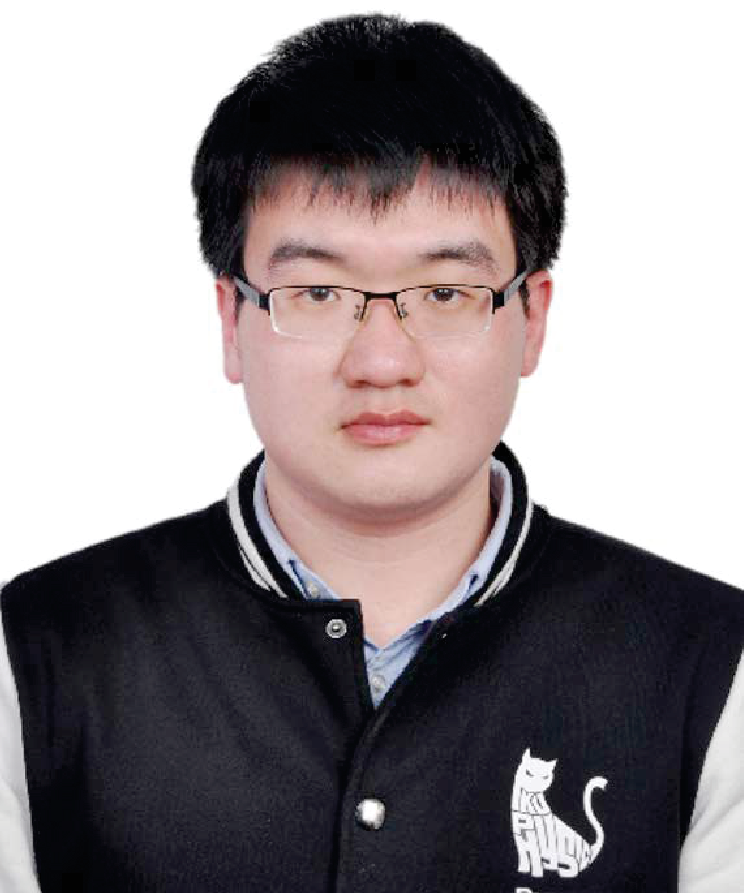}}]{Hongliang Zhang} (S'15) received the B.~S.~degree in Electronic Engineering from Peking University, Beijing, China, in 2014. He is currently pursuing his PhD's degree at School of Electrical Engineering and Computer Science in Peking University. His current research interest includes device-to-device communications, unmanned aerial vehicle networks, and hypergraph theory. He has also served as a TPC Member for GlobeCom 2016, ICC 2016, ICCC 2017, ICC 2018, and GlobeCom 2018.
\end{IEEEbiography}

\begin{IEEEbiography}[{\includegraphics[width=1in,height=1.25in,clip,keepaspectratio]{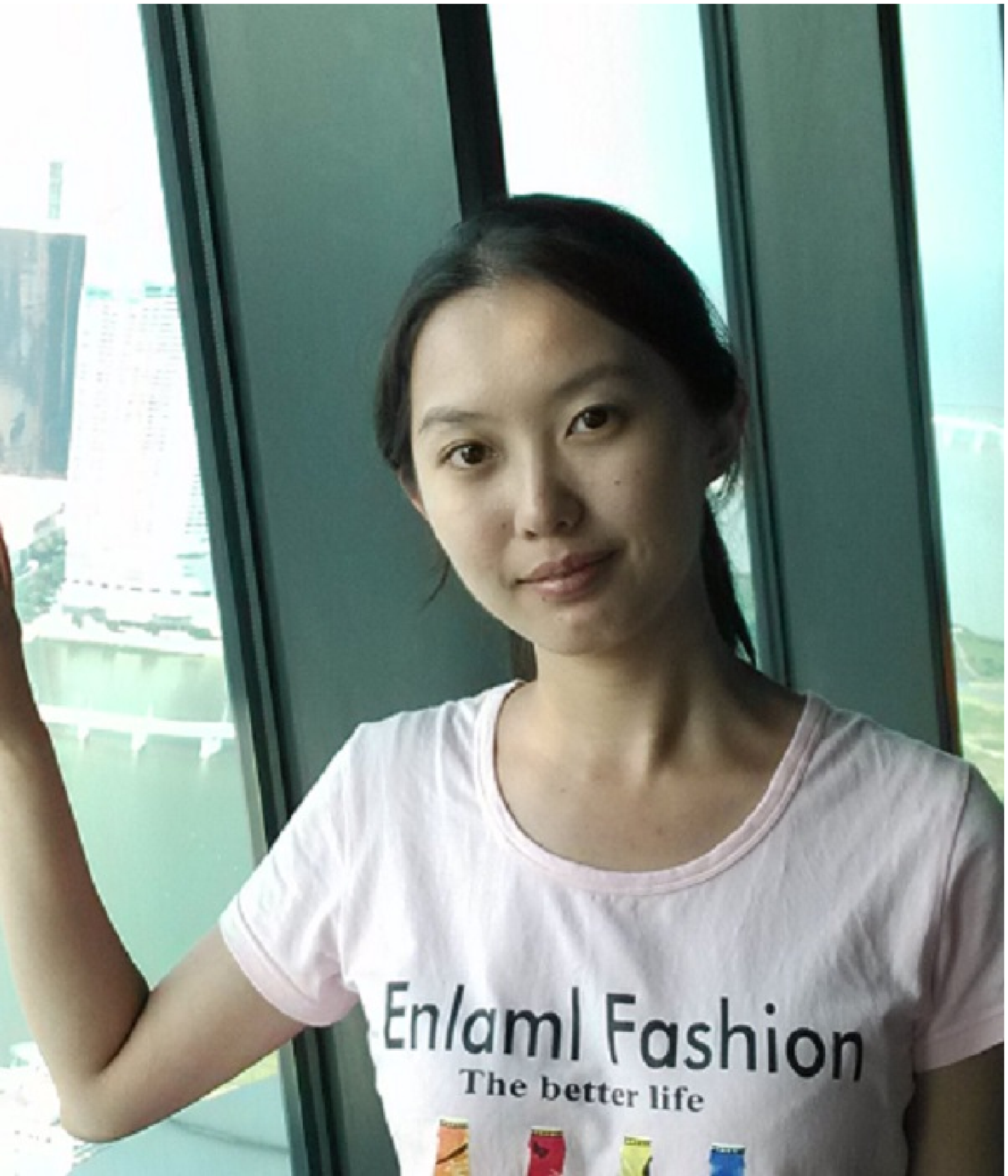}}]{Boya Di}(S'17) received her PhD from Peking University, China, in July 2019. Prior to that, she received the B.S. degree in electronic engineering from Peking University in 2014. Her current research interests include edge computing, random access techniques for 5G, vehicular and UAV networks, and IoT applications. She has also served as a TPC member in GlobeCom 2016, ICCC 2017, ICC 2016 and ICC 2018.
\end{IEEEbiography}

\begin{IEEEbiography}[{\includegraphics[width=1in,height=1.25in,clip,keepaspectratio]{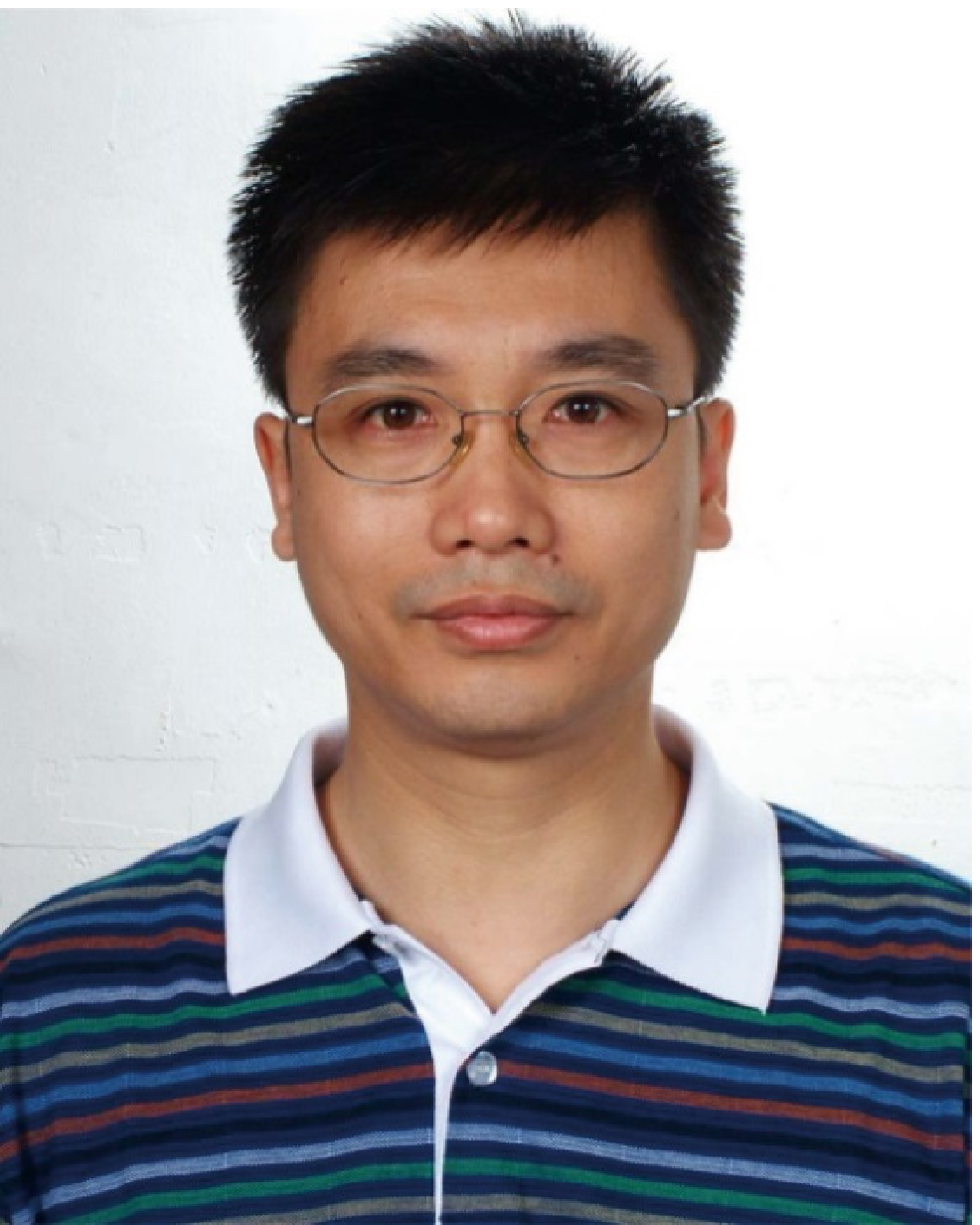}}]{Jianjun Wu} received the B.S., M.S., and Ph.D. degree from Peking University, Beijing, China, in 1989, 1992, and 2006, respectively. Since 1992, he has been with the School of Electronics Engineering and Computer Science, Peking University, and was appointed a Professor in 2014. His research interests are in the areas of satellite communications, wireless communications, and communications signal processing.
\end{IEEEbiography}

\begin{IEEEbiography}[{\includegraphics[width=1in,height=1.25in,clip,keepaspectratio]{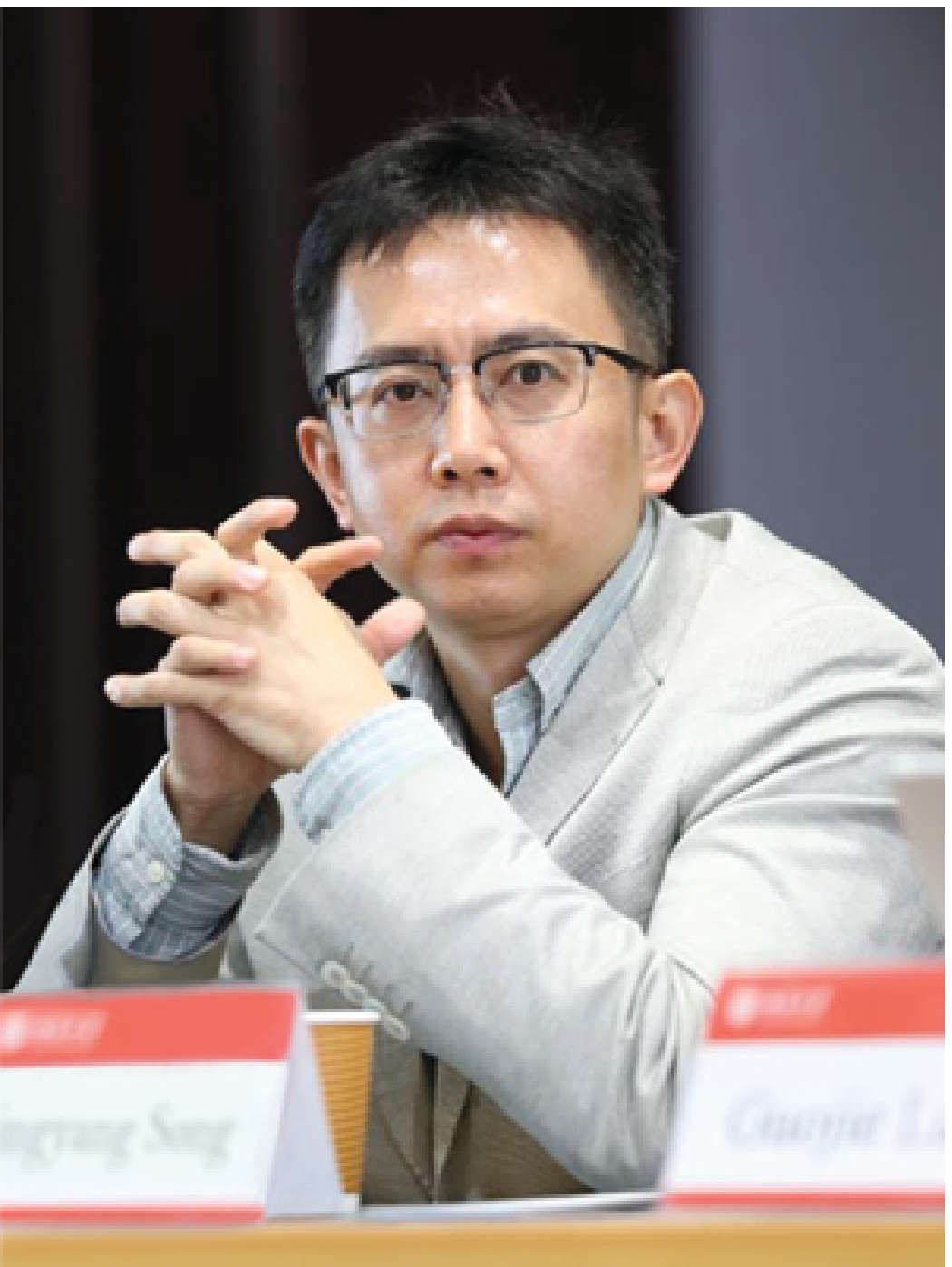}}]{Lingyang Song} (S'03-M'06-SM'12-F'19) received his Ph.D. from the the University of York, United Kingdom, in 2007, where he received the K. M. Stott Prize for excellent research. He worked as a research fellow at the University of Oslo, Norway, until rejoining Philips Research UK in March 2008. In May 2009, he joined the School of Electronics Engineering and Computer Science, Peking University, and is now a Boya Distinguished Professor. His main research interests include wireless communication and networks, signal processing, and machine learning. He was the recipient of the IEEE Leonard G. Abraham Prize in 2016 and the IEEE Asia Pacific (AP) Young Researcher Award in 2012. He has been an IEEE Distinguished Lecturer since 2015.
\end{IEEEbiography}

\end{document}